\def \VersionLong {}
\documentclass[a4paper,10pt]{llncs}
\makeatletter
\AtBeginDocument{%
  \@ifpackageloaded{hyperref}
  {\def\@doi#1{\href{https://doi.org/#1}
      {\ttfamily https://doi.org/#1}\egroup}}
  {\def\@doi#1{\ttfamily https://doi.org/#1\egroup}}
  \def\doi{\bgroup\catcode`\_=12\relax\@doi}}
\makeatother
\makeatletter
\def\@biblabel#1{[#1]}

\makeatother

\usepackage[utf8]{inputenc}
\usepackage[english]{babel}

\usepackage[ruled,vlined,linesnumbered]{algorithm2e}
	\SetKwInOut{Input}{input}
	\SetKwInOut{Output}{output}

\usepackage{subcaption}

\usepackage{paralist} %

\newenvironment{ienumerate}
	{\ifdefined\VersionLong\begin{enumerate}\else\begin{inparaenum}[\itshape i\upshape)]\fi}
	{\ifdefined\VersionLong\end{enumerate}\else\end{inparaenum}\fi}

\usepackage{amsmath} %
\usepackage{amssymb} %

\usepackage[misc,geometry]{ifsym} %

\usepackage{float}

\usepackage{soul} %

\ifdefined \VersionLong
	\newcommand{\LongVersion}[1]{#1}
	\newcommand{\ShortVersion}[1]{}
\else
	\newcommand{\LongVersion}[1]{}
	\newcommand{\ShortVersion}[1]{#1}
\fi

\ifdefined\VersionLong
	\usepackage[backend=biber,backref=true,style=alphabetic,url=false,doi=true,defernumbers=true,sorting=anyt,maxnames=99]{biblatex} %
	\renewbibmacro*{doi+eprint+url}{%
		\iftoggle{bbx:doi}
			{\color{black!40}\footnotesize\printfield{doi}}
			{}%
		\newunit\newblock
		\iftoggle{bbx:eprint}
			{\usebibmacro{eprint}}
			{}%
		\newunit\newblock
		\iftoggle{bbx:url}
			{\usebibmacro{url+urldate}}
			{}%
	}

\fi
\ifdefined\VersionWithComments
	\usepackage{draftwatermark}
	\SetWatermarkText{draft}
	\SetWatermarkScale{15}
	\SetWatermarkColor[gray]{0.9}
\fi
\usepackage[svgnames,table]{xcolor}
\definecolor{darkblue}{rgb}{0, 0, 0.7}

\usepackage[
		pdfauthor={Johan Arcile and Etienne Andre},
		pdftitle={Zone extrapolations in parametric timed automata},
		breaklinks  = true,
		colorlinks  = true,
	\ifdefined \VersionWithComments
	\fi
		citecolor   = blue!50!blue,
		linkcolor   = darkblue,
		urlcolor    = blue!50!green,
	]{hyperref}

\usepackage[capitalise,english,nameinlink]{cleveref} %
\crefname{line}{\text{line}}{\text{lines}} %

\usepackage{tikz}
\usetikzlibrary{shapes,arrows,automata}

\tikzstyle{PTA}=[auto, ->, >=stealth']
\tikzstyle{every node}=[initial text=]
\tikzstyle{location}=[rectangle, rounded corners, minimum size=12pt, draw=black, fill=blue!10, inner sep=2pt]
\tikzstyle{invariant}=[draw=black, dotted, inner sep=1pt] %
\tikzstyle{symbstate}=[state, draw,rectangle,inner sep=3pt]
\tikzstyle{infinitesymbstate}=[symbstate, fill=blue!10]

\tikzstyle{urgent}=[fill=yellow, thick, dotted] %
\tikzstyle{private}=[fill=red,thick]

\definecolor{coloract}{rgb}{0.50, 0.70, 0.30}
\definecolor{colorclock}{rgb}{0.4, 0.4, 1}
\definecolor{colordisc}{rgb}{1, 0, 1}
\definecolor{colorloc}{rgb}{0.4, 0.4, 0.65}
\definecolor{colorparam}{rgb}{1, 0.6, 0.0}

\definecolor{loccolor1}{rgb}{1, 0.3, 0.3}
\definecolor{loccolor2}{rgb}{0.3, 1, 0.3}
\definecolor{loccolor3}{rgb}{0.3, 0.3, 1}
\definecolor{loccolor4}{rgb}{1, 0.3, 1}

\newcommand{\styleact}[1]{\ensuremath{\textcolor{coloract}{{#1}}}}
\newcommand{\styleclock}[1]{\ensuremath{\textcolor{colorclock}{{#1}}}}

\newcommand{\styleparam}[1]{\ensuremath{\textcolor{colorparam}{{#1}}}}

\newcommand{\stylebench}[1]{\texttt{#1}}

\newcommand{\rowHeader}{\rowcolor{blue!20}}

\newcommand{\TO}{T.O.}
\newcommand{\cellTO}{\cellcolor{red!40!black}\textcolor{white}{\TO{}}}

\ifdefined\VersionWithComments
	\usepackage[colorinlistoftodos,textsize=footnotesize]{todonotes}
\else
	\usepackage[disable]{todonotes}
\fi

\ifdefined\VersionWithComments
	\newcommand{\gennote}[3]{\todo[size=\scriptsize,linecolor=#2,backgroundcolor=#2!25,bordercolor=#2]{#3: #1}\xspace}
\else
	\newcommand{\gennote}[3]{}
\fi

\ifdefined \VersionWithComments
	\newcommand{\todoinline}[1]{\mbox{}{\color{red}{\textbf{TODO}\ifx#1\\\else:\ \fi #1}}} %
\else
	\newcommand{\todoinline}[1]{}
\fi

\usepackage{verbatim} %

 \newcommand{\drop}[1]{}

\newcommand{\init}{_0}

\newcommand{\A}{\ensuremath{\mathcal{A}}}
\newcommand{\Actions}{\Sigma}
\newcommand{\action}{\ensuremath{a}}

\newcommand{\assign}{\leftarrow}

\newcommand{\C}{C}

\newcommand{\Clock}{\mathbb{X}} %
\newcommand{\ClockCard}{H} %
\newcommand{\clock}{x} %
\newcommand{\clockx}{\styleclock{x}} %
\newcommand{\clocky}{\styleclock{y}} %
\newcommand{\clockval}{w} %
\newcommand{\ClocksZero}{\vec{0}}
\newcommand{\compOp}{\bowtie}

\newcommand{\constantBoundLUplus}{{\ensuremath{\overline{\LargestP}}}}

\newcommand{\edge}{e}
\newcommand{\Edges}{E}
\newcommand{\longuefleche}[1]{\stackrel{#1}{\longrightarrow}}
\newcommand{\longueflecheRel}[1]{\stackrel{#1}{\mapsto}}

\newcommand{\flecheRel}{{\rightarrow}}

\newcommand{\guard}{g}
\newcommand{\invariant}{I}
\newcommand{\K}{K}
\newcommand{\KFalse}{\bot}
\newcommand{\LargestC}{{\ensuremath{\textcolor{colorok}{M}}}} %
\newcommand{\loc}{\ensuremath{\ell}} %
\newcommand{\locinit}{\loc\init}
\newcommand{\Loc}{L} %
\newcommand{\LocFinal}{\ensuremath{\Loc_{F}}}
\newcommand{\lterm}{\mathit{lt}}
\newcommand{\maxCg}{\ensuremath{C_\mathit{maxg}}} %
\newcommand{\maxC}{\ensuremath{C_\mathit{max}}} %
\newcommand{\overLargestP}{{\ensuremath{\widehat{\LargestP}}}}
\newcommand{\Param}{\mathbb{P}} %
\newcommand{\param}{p} %
\newcommand{\paramp}{\styleparam{p}} %
\newcommand{\ParamCard}{K} %
\newcommand{\PDomain}{\ensuremath{\mathbb{D}}}
\newcommand{\pval}{v} %
\newcommand{\PZG}{\ensuremath{\mathcal{PZG}}} %
\newcommand{\Runs}{\ensuremath{\mathsf{Runs}}}
\newcommand{\sinit}{s\init} %
\newcommand{\somelocs}{T} %
\newcommand{\state}{\ensuremath{s}} %
\newcommand{\States}{S} %
\newcommand{\Succ}{\mathsf{Succ}}
\newcommand{\symtree}{\ensuremath{T^\infty}}
\newcommand{\timelapse}[1]{#1^\nearrow}
\newcommand{\treeprefix}{\ensuremath{\mathit{Tree}}}
\newcommand{\varproblem}{\ensuremath{\varphi}}
\newcommand{\Problem}{\ensuremath{\mathcal{P}}}
\newcommand{\varrun}{\rho} %

\newcommand{\SCGuards}{\mathbb{G}} %
\newcommand{\LargestP}{{\ensuremath{\textcolor{colorok}{N}}}} %
\newcommand{\LargestCL}{{\ensuremath{\textcolor{colorok}{L}}}} %
\newcommand{\LargestCU}{{\ensuremath{\textcolor{colorok}{U}}}} %

\newcommand{\cylinder}[1]{\ensuremath{\textsf{Cyl}_{#1}}}
\newcommand{\Ext}[2]{\ensuremath{\textsf{Ext}^{#1}_{#2}}}

\newcommand{\valuateLU}[1]{\ensuremath{\overline{#1}}}

\newcommand{\setN}{\ensuremath{\mathbb N}}
\newcommand{\setQ}{\ensuremath{{\mathbb Q}}}
\newcommand{\setQplus}{\ensuremath{\setQ_{+}}} %
\newcommand{\setR}{\ensuremath{\mathbb R}}
\newcommand{\setRgeqzero}{\ensuremath{\setR_{\geq 0}}}

\newcommand{\setRplus}{\ensuremath{\setR_{+}}} %
\newcommand{\setZ}{\ensuremath{\mathbb Z}}

\newcommand{\styleSymbStatesSet}[1]{\ensuremath{\mathbf{#1}}}
\newcommand{\Passed}{\styleSymbStatesSet{P}}
\newcommand{\symbstate}{\ensuremath{\styleSymbStatesSet{s}}} %
\newcommand{\SymbState}{\ensuremath{\styleSymbStatesSet{S}}} %
\newcommand{\symbstateinit}{\symbstate\init} %
\newcommand{\symbtrans}{{\Rightarrow}} %

\newcommand{\resets}{R}

\newcommand{\projectP}[1]{\ensuremath{#1{\downarrow_{\Param}}}}
\newcommand{\reset}[2]{\ensuremath{[#1]_{#2}}}
\newcommand{\valuate}[2]{\ensuremath{#2(#1)}}
\newcommand{\wv}[2]{#1|#2} %

\newcommand{\stylealgo}[1]{\ensuremath{\mathsf{#1}}}
\newcommand{\EEF}{\stylealgo{EEF}}
\newcommand{\EEFbar}{\ensuremath{\overline{\mathsf{E}}\stylealgo{EF}}}
\newcommand{\EEFhat}{\ensuremath{\widehat{\mathsf{E}}\stylealgo{EF}}}
\newcommand{\EEFp}{\ensuremath{\stylealgo{pEEF}}}
\newcommand{\EEFvect}{\ensuremath{\vec{\mathsf{E}}\stylealgo{EF}}}
\newcommand{\EF}{\stylealgo{EF}} %

\newcommand{\propEF}{\stylebench{reach}}
\newcommand{\propCycle}{\stylebench{liveness}}
\ifdefined\VersionLong
	\newcommand{\propAGnot}{\stylebench{safety}}
	\newcommand{\propAccCycle}{\stylebench{acc liveness}}
\else
	\newcommand{\propAGnot}{\propEF{}}
	\newcommand{\propAccCycle}{\propCycle{}}
\fi

\newcommand{\imitator}{\textsf{IMITATOR}}
\newcommand{\uppaal}{\textsc{Uppaal}}

\ifdefined \VersionWithComments
 	\definecolor{colorok}{RGB}{80,80,150}
\else
	\definecolor{colorok}{RGB}{0,0,0}
\fi

\newcommand{\eg}{\textcolor{colorok}{e.g.,}\xspace}
\newcommand{\etal}{\textcolor{colorok}{\emph{et~al.}}\xspace}
\newcommand{\ie}{\textcolor{colorok}{i.e.,}\xspace}
\newcommand{\suchthat}{\textcolor{colorok}{s.t.}\xspace}

\makeatletter
\def\orcidID#1{\smash{\href{https://orcid.org/#1}{\protect\raisebox{-1.25pt}{\protect\includegraphics{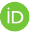}}}}}
\makeatother

\title{Zone extrapolations in parametric timed automata\todo{This is the version with comments. To disable comments, comment out line~3 in the \LaTeX{} source.}\thanks{%
	\LongVersion{%
		This is the author (and extended) version of the manuscript of the same name published in the proceedings of the 14th {NASA} Formal Methods Symposium (\href{https://nfm2022.caltech.edu/}{NFM 2022}).
		The final authenticated version is available at %
			\href{https://www.springer.com}{\nolinkurl{springer.com}}.
	}%
	This work is partially supported by the ANR-NRF French-Singaporean research program \href{https://www.loria.science/ProMiS/}{ProMiS} (ANR-19-CE25-0015).
}
}

\author{%
	Johan Arcile%
${}^{\text{\Letter}}$
\and
	\'Etienne Andr\'e%
	\orcidID{0000-0001-8473-9555}
}
\institute{Université de Lorraine, CNRS, Inria, LORIA, F-54000 Nancy, France}

\begin{document}
\sloppy

\pagestyle{plain}

\maketitle

\thispagestyle{plain}

\begin{abstract}
	Timed automata (TAs) are an efficient formalism to model and verify systems with hard timing constraints, and concurrency.
	While TAs assume exact timing constants with infinite precision, parametric TAs (PTAs) leverage this limitation and increase their expressiveness, at the cost of undecidability.
	A practical explanation for the efficiency of TAs is zone extrapolation, where clock valuations beyond a given constant are considered equivalent.
	This concept cannot be easily extended to PTAs, due to the fact that parameters can be unbounded. %
	In this work, we propose several definitions of extrapolation for PTAs based on the $\LargestC$-extrapolation, and we study their correctness.
	Our experiments show an overall decrease of the computation time and, most importantly, allow termination of some previously unsolvable benchmarks.
	
\end{abstract}
\section{Introduction}\label{section:introduction}

Timed automata (TAs)~\cite{AD94} represent an efficient and expressive formalism to model and verify systems mixing hard timing constraints with concurrency, being one of the most expressive decidable formalisms with timing constraints.
However, TAs assume exact timing constants with infinite precision, which may not be realistic in practice; in addition, they assume full knowledge of the model, preventing verification at an early development phase.
Parametric timed automata (PTAs) leverage these limitations, by allowing unknown timing constants in the model---at the cost of undecidability: the mere emptiness of the parameter valuations set for which a given (discrete) location is reachable (called \emph{reachability emptiness}) is undecidable~\cite{AHV93}.

A practical explanation for the efficiency of TAs for reachability properties is \emph{(zone) extrapolation}, where clock valuations beyond a given constant are considered to be equivalent.
Since the seminal work~\cite{AD94}, several works improved the quality and efficiency of zone extrapolation, by considering different constants per clock~\cite{BBFL03,BBLP06} or extending extrapolation to liveness properties~\cite{Tripakis09,Li09}.
This concept cannot be easily extended to PTAs, due to the fact that parameters can be unbounded, 
\LongVersion{or that one of their bound may converge towards a constant (for example $\frac{1}{n} \leq p$, with $n$ growing without bound.).}%
\ShortVersion{or converge toward infinitely small values.}

\LongVersion{
\subsection{Related works}
}

\paragraph{Extrapolation in TAs}
Daw and Tripakis first introduced the \emph{extrapolation} abstraction in \cite{DT98} as a mean to obtain a finite simulation of the state space of TAs.
The extrapolation abstraction preserves reachability properties and is based on the largest constant appearing in any state of the model, which can be computed syntactically from the constants present in its guards and invariants.
In~\cite{BBFL03} Behrmann \etal{} redefine this abstraction with individual clock bounds (\ie{} the largest constant is computed for each clock) and will later refer to it in~\cite{BBLP06} as the $\LargestC$-extrapolation.
\LongVersion{%
	In this latter work~\cite{BBLP06}, the $\LargestC$-extrapolation is extended to a coarser abstraction based on two constants for each clock: its greater lower bound and its greater upper bound.
This new form of extrapolation is referred to as the $\LargestCL\LargestCU$-extrapolation and still preserves reachability properties.
}%
Experiments are performed using \uppaal{}~\cite{LPY97}.
\LongVersion{In 2009, }Tripakis~\cite{Tripakis09} showed that the $\LargestC$-extrapolation is correct for checking emptiness of timed Büchi automata, \ie{} checking for accepting cycles in~TAs.
\LongVersion{The same year, Li~\cite{Li09} proves that this result holds true for the $\LargestCL\LargestCU$-extrapolation on TAs.}

\paragraph{Parameter synthesis for PTAs}
Most non-trivial decision problems are undecidable for PTAs (see \cite{Andre19STTT} for a survey).
As a consequence exact synthesis is usually out of reach, except for small numbers of clocks or of parameters (see, \eg{} \cite{AHV93,BBLS15,BO17}).
For general subclasses (without bound on the number of variables), exact synthesis results are very scarce.
Some fit in the subclasses of L/U-PTAs\footnote{While ``L/U'' means in both cases ``lower-upper (bound)'', L/U-PTAs are a completely different concept from LU-extrapolation for (P)TAs.}~\cite{HRSV02}, and notably in U-PTAs (resp.\ L-PTAs)~\cite{BlT09}, where each timing parameter is constrained to be always compared to a clock as an upper (resp.\ lower) bound, \ie{} of the form $\clock \leq \param$ (resp.\ $\param \leq \clock$).
The only known situations when exact reachability-synthesis (\ie{} synthesis of all parameter valuations for which a given location is reachable) can be achieved for subclasses of PTAs are
\begin{ienumerate}%
	\item \LongVersion{reachability-synthesis }for U-PTAs (resp.\ L-PTAs) over \emph{integer-valued} timing parameters~\cite{BlT09};
	\item \LongVersion{reachability-synthesis }for the whole PTA class, over \emph{bounded and integer-valued} parameters (which reduces to TAs)~\cite{JLR15};
		and
	\item \LongVersion{reachability-synthesis }for reset-update-to-parameters-PTAs (``R-U2P-PTAs''), in which all clocks must be updated (possibly to a parameter) whenever a clock is compared to a parameter in a guard~\cite{ALR21}.
\end{ienumerate}%
On the negative side, even L/U-PTAs show negative results for synthesis: while reachability-emptiness is decidable for L/U-PTAs~\cite{HRSV02}, reachability-synthesis is intractable (its result cannot be represented using a finite union of polyhedra)~\cite{JLR15}; and even in the very restricted subclass of U-PTAs without invariant, TCTL-emptiness (\ie{} \LongVersion{the }emptiness of the parameter valuations set for which a TCTL formula is valid) is undecidable~\cite{ALR18FORMATS}.

We performed a first attempt to define an extrapolation for PTAs in~\cite{ALR15}: we adapted the $\LargestC$-extrapolation to the context of PTAs, although restricted to \emph{bounded} parameter domains only.
No implementation was provided.
In~\cite{BBBC16}, the authors also define an extrapolation very similar to~\cite{ALR15}.
Compared to~\cite{ALR15}, we reuse here some of the definitions of~\cite{ALR15}, and we significantly extend the definition of extrapolations; we also consider several subclasses of models, as well as liveness properties; we also perform an experimental evaluation.

\LongVersion{%
\subsection{Contributions}
}\ShortVersion{\paragraph{Contributions}}

We propose several definitions of extrapolation for PTAs, and study their correctness.
In the context of bounded parameter domains, we extend the parametric $\LargestC$-extrapolation from~\cite{ALR15} to individual clock bounds.
Those extrapolations are combined with results from~\cite{BlT09} to cope with the issue raised by unbounded parameters.
We notably consider variants of the U-PTAs and L-PTAs.
We show that, on the subclass of (unbounded) PTAs on which they apply, those abstractions preserve not only reachability-synthesis but also cycle-synthesis (``liveness'').
We perform experiments using \LongVersion{the parametric timed model checker }\imitator{}~\cite{Andre21}, including on the most general class (rational-valued, possibly unbounded parameters).
With the aforementioned negative theoretical results in mind, our evaluation focuses on evaluating the speed enhancement, and the increase of termination chances for our case studies.
We show that, overall, extrapolation decreases the verification time and, most importantly, can effectively solve previously unsolvable benchmarks.

\paragraph{Outline}
We introduce the necessary preliminaries in \cref{section:preliminaries}.
The $\LargestC$-extrapolation in the bounded context (partially reusing results from~\cite{ALR15}) is studied in \cref{section:M-Ex_bounded}.
\cref{section:M-Ex_unbounded} adapts the $\LargestC$-extrapolation to the unbouded context for reachability properties.
Liveness %
properties are discussed in \cref{section:beyond}.
Finally, \cref{section:experiments} benchmarks the abstractions, and \cref{section:conclusion} concludes the paper.

\section{Preliminaries}\label{section:preliminaries}

\LongVersion{%
\subsection{Clocks, parameters and guards}
}

Throughout this paper, we assume a set~$\Clock = \{ \clock_1, \dots, \clock_\ClockCard \} $ of \emph{clocks}, \ie{} real-valued variables that evolve at the same rate.
A clock valuation is a function $\clockval : \Clock \rightarrow \setRplus$.
\LongVersion{%
	We identify a clock valuation $\clockval$ with the \emph{point} $(\clockval(\clock_1), \dots, \clockval(\clock_\ClockCard))$.
}%
We write $\ClocksZero$ for the clock valuation assigning $0$ to all clocks.
Given $d \in \setRplus$, $\clockval + d$ denotes the valuation \suchthat{} $(\clockval + d)(\clock) = \clockval(\clock) + d$, for all $\clock \in \Clock$.
Given $\resets \subseteq \Clock$, we define the \emph{reset} of a valuation~$\clockval$, denoted by $\reset{\clockval}{\resets}$, as follows: $\reset{\clockval}{\resets}(\clock) = 0$ if $\clock \in \resets$, and $\reset{\clockval}{\resets}(\clock)=\clockval(\clock)$ otherwise.

We assume a set~$\Param = \{ \param_1, \dots, \param_\ParamCard \} $ of \emph{parameters}, \ie{} unknown constants.
A parameter {\em valuation} $\pval$ is a function $\pval : \Param \rightarrow \setQ$.
\LongVersion{%
	We identify a valuation $\pval$ with the \emph{point} $(\pval(\param_1), \dots, \pval(\param_\ParamCard))$.
}%
Given two valuations $\pval_1, \pval_2$, we write $\pval_1 \geq \pval_2$ whenever $\forall \param \in \Param$, $\pval_1(\param) \geq \pval_2(\param)$.

In the following, we assume ${\compOp} \in \{<, \leq, =, \geq, >\}$.
A \emph{constraint}~$\C$ over $\Clock \cup \Param$ is a conjunction of inequalities of the form $\lterm \compOp 0$, where $\lterm$ is a linear term over $\Clock \cup \Param$ of the form $\sum_{1 \leq i \leq \ClockCard} \alpha_i \clock_i + \sum_{1 \leq j \leq \ParamCard} \beta_j \param_j + d$, with $\clock_i \in \Clock$, $\param_j \in \Param$, and $\alpha_i, \beta_j, d \in \setZ$.
We also refer to constraints as their geometrical representation, \ie{} of \emph{convex polyhedron}.
\LongVersion{

}%
We denote by~$\KFalse$ the constraint over~$\Param$ corresponding to the empty set of parameter valuations.

Given a parameter valuation~$\pval$, $\valuate{\C}{\pval}$ denotes the constraint over~$\Clock$ obtained by replacing each parameter~$\param$ in~$\C$ with~$\pval(\param)$.
Likewise, given a clock valuation~$\clockval$, $\valuate{\valuate{\C}{\pval}}{\clockval}$ denotes the expression obtained by replacing each clock~$\clock$ in~$\valuate{\C}{\pval}$ with~$\clockval(\clock)$.
We say that $\pval$ \emph{satisfies}~$\C$, denoted by $\pval \models \C$, if the set of clock valuations satisfying~$\valuate{\C}{\pval}$ is nonempty.
Given a parameter valuation $\pval$ and a clock valuation $\clockval$, 
we denote by $\wv{\clockval}{\pval}$ the valuation over $\Clock\cup\Param$ such that for all clocks $\clock$, 
$\valuate{\clock}{\wv{\clockval}{\pval}}=\valuate{\clock}{\clockval}$ and for all parameters $\param$, 
$\valuate{\param}{\wv{\clockval}{\pval}}=\valuate{\param}{\pval}$.
We use the notation $\wv{\clockval}{\pval} \models \C$ to indicate that $\valuate{\valuate{\C}{\pval}}{\clockval}$ evaluates to true.
We say that $\C$ is \emph{satisfiable} if $\exists \clockval, \pval \text{ s.t.\ } \wv{\clockval}{\pval} \models \C$.

We define the \emph{time elapsing} of~$\C$, denoted by $\timelapse{\C}$, as the constraint over $\Clock$ and $\Param$ obtained from~$\C$ by delaying all clocks by an arbitrary amount of time.
That is, \(\wv{\clockval'}{\pval} \models \timelapse{\C} \text{ iff } \exists \clockval : \Clock \to \setRplus, \exists d \in \setRplus \text { s.t. } \wv{\clockval}{\pval} \models \C \land \clockval' = \clockval + d \text{.}\)

Given $\resets \subseteq \Clock$, we define the \emph{reset} of~$\C$, denoted by $\reset{\C}{\resets}$, as the constraint obtained from~$\C$ by resetting the clocks in~$\resets$, and keeping the other clocks unchanged.
We denote by $\projectP{\C}$ the projection of~$\C$ onto~$\Param$, \ie{} obtained by eliminating the variables not in~$\Param$ (\eg{} using Fourier-Motzkin~\cite{Schrijver86}).

A \emph{simple clock guard} is an inequality of the form $\clock \compOp \sum_{1 \leq i \leq \ParamCard} \alpha_i \param_i + z$, with $\param_i \in \Param$, and $\alpha_i, z \in \setZ$.
A \emph{clock guard} is a constraint over $\Clock \cup \Param$ defined by a conjunction of simple clock guards.
Given a clock guard~$\guard$, we write~$\clockval\models\pval(\guard)$ if the expression obtained by replacing each~$\clock$ with~$\clockval(\clock)$ and each~$\param$ with~$\pval(\param)$ in~$\guard$ evaluates to true.
\LongVersion{%
	We do not consider diagonal constraints (\ie{} simple clock guards of the form $\clock - \clock' \compOp …$)\ in this work.
}

\LongVersion{%
\subsection{Parametric timed automata}
}%
\ShortVersion{\paragraph{PTAs}}
Parametric timed automata (PTAs) extend timed automata with parameters within guards and invariants in place of integer constants~\cite{AHV93}.

\begin{definition}[PTA]\label{def:PTA}
	A PTA $\A$ is a tuple \mbox{$\A = (\Actions, \Loc, \locinit, \LocFinal, \Clock, \Param, \PDomain, \invariant, \Edges)$}, where:
	\begin{ienumerate}
		\item $\Actions$ is a finite set of actions,
		\item $\Loc$ is a finite set of locations,
		\item $\locinit \in \Loc$ is the initial location,
		\item $\LocFinal \subseteq \Loc$ is a set of accepting locations,
		\item $\Clock$ is a finite set of clocks,
		\item $\Param$ is a finite set of parameters,
		\item $\PDomain : \Param \rightarrow (\setQ \cup \{ - \infty \}) \times (\setQ \cup \{ + \infty \}) $ is the parameter domain,
		\item $\invariant$ is the invariant, assigning to every $\loc\in \Loc$ a clock guard $\invariant(\loc)$,
		\item $\Edges$ is a finite set of edges  $\edge = (\loc,\guard,\action,\resets,\loc')$
		where~$\loc,\loc'\in \Loc$ are the source and target locations, $\action \in \Actions$, $\resets\subseteq \Clock$ is a set of clocks to be reset, and $\guard$ is a clock guard.
	\end{ienumerate}
\end{definition}

Let $\SCGuards(\A)$ denote the set of all simple clock guards of the PTA~$\A$, \ie{} all simple clock guards being a conjunct within a guard or an invariant of~$\A$.
Given a clock $\clock \in \Clock$, we denote by $\SCGuards^\clock(\A) \subseteq \SCGuards(\A)$ the set of simple clock guards where $\clock$ appears, \ie{} is bound by a non-0 coefficient.
A clock~$\clock$ of~$\A$ is said to be a \emph{parametric clock} if it is compared to at least one parameter (with a non-0 coefficient) in at least one guard of $\SCGuards^\clock(\A)$.

The parameter domain of a PTA is the admissible range of the parameters.
Given~$\param$, given $\PDomain(\param) = (b^-, b^+)$, $\PDomain^-(\param)$ denotes $b^-$ while $\PDomain^+(\param)$ denotes $b^+$.
The admissible valuations for~$\param$ are therefore $[\PDomain^-(\param) , \PDomain^+(\param)]$ (the domain is \emph{closed} unless on the side of an infinite bound).
A \emph{bounded} parameter domain assigns to each parameter a minimum \LongVersion{rational bound }and a maximum rational bound.
In that case,
$\PDomain^-(\param_i) > -\infty$ and $\PDomain^+(\param_i) < +\infty$.
A bounded parameter domain can be seen as a hyperrectangle in $\ParamCard$ dimensions.
Any parameter that is not bounded is \LongVersion{called an }\emph{unbounded}\LongVersion{ parameter}.
Note that an unbounded parameter can still have a lower bound or an upper bound $\in \setQ$.
\LongVersion{

\begin{definition}[bounded PTA]\label{def:boundedPTA}
	A \emph{bounded PTA} is a PTA the parameter domain of which is bounded.
	Otherwise, it is \emph{unbounded}.
\end{definition}
}%
\ShortVersion{A PTA is \emph{bounded} if its parameter domain is bounded; otherwise, it is \emph{unbounded}.}

Given a parameter valuation~$\pval$, we denote by $\valuate{\A}{\pval}$ the non-parametric structure where all occurrences of a parameter~$\param_i$ have been replaced by~$\pval(\param_i)$.
We denote as a \emph{timed automaton} any structure $\valuate{\A}{\pval}$, by assuming a rescaling of the constants: by multiplying all constants in $\valuate{\A}{\pval}$ by the least common multiple of their denominators, we obtain an equivalent (integer-valued) TA\LongVersion{, as defined in}~\cite{AD94}.

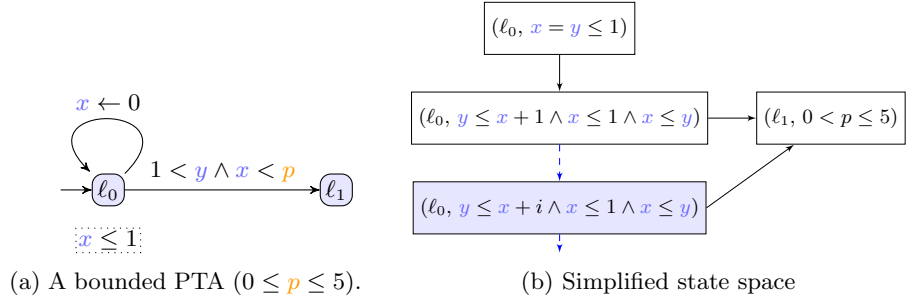
\begin{figure} [tb]
	\centering
	\small
	\begin{subfigure}[b]{0.38\textwidth}
		\centering
		\begin{tikzpicture}[PTA, thin]

			\node[location, initial] at(0,0) (l0) {$\loc_0$};
			\node[location] at(3,0) (l1) {$\loc_1$};
			\node [invariant, below] at (0,-0.5) {$\clockx \leq 1$};

			\path
				(l0) edge[loop] node [above] {$\clockx \assign 0$} (l0)
				(l0) edge[] node {$1< \clocky \land \clockx < \paramp$}  (l1)
			;
		\end{tikzpicture}
		\caption{A bounded PTA ($0 \leq \paramp \leq 5$).}
		\label{figure:bounded_pta:pta}
	\end{subfigure}
		\hfill
	\begin{subfigure}[b]{0.6\textwidth}
		\centering
		\scalebox{.85}{
		\begin{tikzpicture}[>=latex', xscale=1.2, yscale=.7,every node/.style={scale=1}]
		\node[symbstate] at (0,0) (c0) {($\loc_0$, $\clockx = \clocky \leq 1$)};
		\node[symbstate] at (0,-2) (c1) {($\loc_0$, $\clocky \leq \clockx + 1 \land \clockx \leq 1 \land \clockx \leq \clocky$)};
		\node[symbstate] at (3.5,-2) (c1') {($\loc_1$, $0 < \param \leq 5$)};
		\node[infinitesymbstate] at (0,-4) (ci) {($\loc_0$, $\clocky \leq \clockx + i \land \clockx \leq 1 \land \clockx \leq \clocky$)};
		\draw[->] (c0) -- (c1);
		\draw[->] (c1) -- (c1');
		\draw[->, blue,dashed] (c1) -- (ci);
		\draw[->] (ci.east) -- (c1');
		\draw[->, blue,dashed] (ci) -- (0,-5);
		\end{tikzpicture}
		}
		\caption{Simplified state space}
		\label{figure:bounded_pta:ss}
	\end{subfigure}
	\caption{
		\ShortVersion{A bounded PTA and its infinite state space.}
		\LongVersion{Example of a bounded PTA generating an infinite state space. Blue states are a representation of an infinite sequence of states where variable $i$ corresponds to the number of times the looping transition on $\loc_0$ was taken.}
		}
	\label{figure:bounded_pta}
\end{figure}

\begin{example}
	\cref{figure:bounded_pta:pta} displays \LongVersion{the graphical representation of }a bounded PTA.
	We have
		$\SCGuards(\A) = \{ \clockx \leq 1, 1< \clocky, \clockx < \paramp \}$,
		$\SCGuards^{\clockx}(\A) = \{ \clockx \leq 1, \clockx < \paramp \}$,
		and
		$\SCGuards^{\clocky}(\A) = \{ 1< \clocky \}$.
	The valuation of\LongVersion{ parameter}~$\param$ can be any rational value in $[0,5]$\LongVersion{, hence an infinite number of possible parameter valuations}.
	Therefore, this PTA can be seen as the abstract representation for an infinite number of TAs.
\end{example}

\LongVersion{%
\subsubsection{Concrete semantics of TAs}

Let us now recall the concrete semantics of TAs.
}
\begin{definition}[Semantics of a TA]
	Given a PTA $\A = (\Actions, \Loc, \locinit, \LocFinal, \Clock, \Param, \PDomain, \invariant, \Edges)$ and a parameter valuation~\(\pval\), 
	the concrete semantics of $\valuate{\A}{\pval}$ is given by the timed transition system $(\States, \sinit, \flecheRel)$, with
	\begin{itemize}
		\item $\States = \{ (\loc, \clockval) \in \Loc \times \setRgeqzero^\ClockCard \mid \clockval \models \valuate{\invariant(\loc)}{\pval} \}$,
		\LongVersion{\item }$\sinit = (\locinit, \ClocksZero) $,
		\item  $\flecheRel$ consists of the (continuous) delay and discrete transition relations:
		\begin{itemize}
			\item delay transitions: $(\loc, \clockval) \longueflecheRel{d} (\loc, \clockval+d)$, with $d \in \setRgeqzero$, if $\forall d' \in [0, d], (\loc, \clockval+d') \in \States$;
			\item discrete transitions: $(\loc, \clockval) \longueflecheRel{\edge} (\loc',\clockval')$, if $(\loc, \clockval) , (\loc',\clockval') \in \States$, and there exists $\edge = (\loc,\guard,\action,\resets,\loc') \in \Edges$, such that $\clockval'= \reset{\clockval}{\resets}$, and $\clockval\models\pval(\guard$).
		\end{itemize}
	\end{itemize}
\end{definition}

Moreover, we write $(\loc, \clockval)\longuefleche{(d, \edge)} (\loc',\clockval')$ for a combination of a delay and discrete transition if $\exists  \clockval'' :  (\loc, \clockval) \longueflecheRel{d} (\loc, \clockval'') \longueflecheRel{\edge} (\loc',\clockval')$.

Given a TA~$\valuate{\A}{\pval}$ with concrete semantics $(\States, \sinit, \flecheRel)$, we refer to the states of~$\States$ as the \emph{concrete states} of~$\valuate{\A}{\pval}$.
A \emph{run} of~$\valuate{\A}{\pval}$ is an alternating sequence of concrete states of $\valuate{\A}{\pval}$ and pairs of edges and delays starting from the initial state $\sinit$ 
and is of the form $\state_0, (d_0, \edge_0), \state_1, \cdots \state_i, (d_i, \edge_i), \cdots$ with $i = 0, 1, \dots$, $\edge_i \in \Edges$, $d_i \in \setRgeqzero$ and $\state_i \longuefleche{(d_i, \edge_i)} \state_{i+1}$.
The set of all (finite or infinite) runs of a TA~$\valuate{\A}{\pval}$ is $\Runs(\valuate{\A}{\pval})$.
Given a concrete state $\state = (\loc, \clockval)$, we say that $\state$ is reachable in~$\valuate{\A}{\pval}$ (and by extension that $\loc$ is reachable, or that $\valuate{\A}{\pval}$ visits $\loc$) if $\state$ appears in a run of $\valuate{\A}{\pval}$.
An infinite run is \emph{accepting} if it visits infinitely often (at least) one location~$\loc \in \LocFinal$.
\subsubsection{Symbolic semantics of PTAs}\label{ss:symbolic}

Let us now recall the symbolic semantics of PTAs (see \eg{} \cite{HRSV02,ACEF09}).
\LongVersion{

\begin{definition}[Symbolic state]
	}A symbolic state is a pair $(\loc, \C)$ where $\loc \in \Loc$ is a location, and $\C$ is a constraint over $\Clock \cup \Param$ called its associated \emph{parametric zone}.
\LongVersion{\end{definition}}
\begin{definition}[Symbolic semantics]\label{def:PTA:symbolic}
	Given a PTA $\A = (\Actions, \Loc, \locinit, \LocFinal, \Clock, \Param, \PDomain, \invariant, \Edges)$,
	the symbolic semantics of~$\A$ is the labeled transition system called \emph{parametric zone graph}
	$ \PZG = ( \Edges, \SymbState, \symbstateinit, \symbtrans )$, with
	\begin{itemize}
		\item $\SymbState = \{ (\loc, \C) \mid \C \subseteq \invariant(\loc) \}$, 
		\LongVersion{\item }$\symbstateinit = \big(\locinit, \timelapse{(\bigwedge_{1 \leq i\leq\ClockCard}\clock_i=0)} \land \invariant(\loc_0) \land \bigwedge_{1 \leq j\leq \ParamCard} \PDomain^-(\param_j) \leq \param_j \leq \PDomain^+(\param_j) \big)$,
				and
		\item $\big((\loc, \C), \edge, (\loc', \C')\big) \in \symbtrans $ if $\edge = (\loc,\guard,\action,\resets,\loc') \in \Edges$ and
			\(\C' = \timelapse{\big(\reset{(\C \land \guard)}{\resets}\land \invariant(\loc')\big )} \land \invariant(\loc') \text{, with $\C'$ satisfiable.}\) 
	\end{itemize}
\end{definition}

That is, in the parametric zone graph, nodes are symbolic states, and arcs are labeled by \emph{edges} of the original PTA.
Given  $(\symbstate, \edge, \symbstate') \in \symbtrans $, we write $\symbstate' = \Succ(\symbstate, \edge)$.
\LongVersion{

}%
Given a concrete state $\state = (\loc, \clockval)$ and a symbolic state $\symbstate = (\loc', \C)$, we write
$\state \in \symbstate$ whenever $\loc = \loc'$ and $\clockval \models \C$.

\begin{example}
	\cref{figure:bounded_pta:ss} displays the parametric zone graph of the PTA in \cref{figure:bounded_pta:pta}.
	Blue states represent an infinite sequence ($i$ being the number of times the looping transition was taken).
	\LongVersion{(We assume all clocks and parameters to be non-negative and, for sake of brevity, constraints of the form $\clock \geq 0$ may be omitted.)}
\end{example}
\subsubsection{Computation problems}

Given a class of decision problems \Problem{} (reachability, \LongVersion{liveness, }etc.), we consider the problem of synthesizing the set (or part of it) of parameter valuations $\pval$ such that $\valuate{\A}{\pval}$ satisfies $\varproblem$.
Here, we mainly focus on reachability (\ie{} ``does there exist a run that reaches some given location?'') and liveness (\ie{} ``does there exist a run that visits a given location infinitely often?'').
\section{$\LargestC$- and $\vec{\LargestC}$-extrapolation for bounded PTAs} \label{section:M-Ex_bounded}
\subsection{Recalling $\LargestC$-extrapolation}\label{ss:Mextrapolation}

In this subsection, we recall some results from~\cite{BBLP06,ALR15}, where the classical ``$k$-extrapolation'' used for the zone-abstraction of TAs is adapted to PTAs.
While this part is not clearly a contribution of the current manuscript, we redefine some concepts from~\cite{ALR15}, and provide several original examples.

\LongVersion{%
\subsubsection{Maximal constant of a bounded PTA}
}
\LongVersion{First, let us formally define the \emph{maximal constant} of a bounded PTA.}
The maximal constant~$\LargestC$ is the maximum value that can appear in the guards and invariants of the PTA.
When those constraints are parametric expressions, we compute the maximum value that the expression can take over any parameter valuation within the (bounded) parameter domain~$\PDomain$ (this maximal value is unique since expressions are linear).

Given a simple clock guard $\guard$ of the form $\clock \compOp \sum_{1 \leq i \leq \ParamCard} \alpha_i \param_i + z$ %
we define $\maxCg(\guard) = \sum_{1 \leq i \leq \ParamCard} \alpha_i \gamma_i + z$ where
\begin{ienumerate}
	\item $\gamma_i = \PDomain^-(\param_i)$ if $\alpha_i < 0$,
	\item $\gamma_i = \PDomain^+(\param_i)$ if $\alpha_i > 0$, and
	\item $\gamma_i = 0$ otherwise.
\end{ienumerate}

\begin{example}
	Consider the simple clock guard~$\guard : \clock\leq 2\param_1 -\param_2+1$ and $\param_1\in [2,5]$, and $\param_2\in [-3,4]$; then $\maxCg(\guard) = 2 \times 5 - (-3) + 1 = 14$.
\end{example}
\begin{definition}[Maximal constant]\label{def:maximalC}
	Given a bounded PTA~$\A$, for any clock~$\clock \in \Clock$,
	the maximal constant for clock~$\clock$ is
	\(\maxC^{\clock}(\A) = \max_{\guard \in \SCGuards^{\clock}(\A)} \maxCg(\guard)\)
	furthermore, the maximal constant of the PTA is
	\(\maxC(\A) = \max_{\guard \in \SCGuards(\A)} \maxCg(\guard)\text{.}\)
\end{definition}
\begin{example}
	\ShortVersion{Consider again \cref{figure:bounded_pta:pta}.
	Then, $\maxC^{\clockx}(\A)=5$ and $\maxC^{\clocky}(\A)=1$.}
	\LongVersion{Consider again \cref{figure:bounded_pta:pta} (recall that $0 \leq \paramp \leq 5$).
	Then, $\maxC^{\clockx}(\A) = 5$, $\maxC^{\clocky}(\A) = 1$ and $\maxC(\A) = 5$.}
\end{example}

\LongVersion{%
\subsubsection{Bisimulation and largest constant in TAs}
}
Let us recall from~\cite{BBLP06} the notion of bisimulation based on\LongVersion{ the maximal constant}~$\LargestC$:

\begin{lemma}[{\cite[Lemma~1]{BBLP06}}]\label{lemma:bisim}
	Let~$\A$ be a TA.
	Given clock~$\clock$, let $\LargestC(\clock)$ be an integer constant greater than or equal to $\maxC^{\clock}(\A)$.
	Let $\clockval, \clockval'$ be two clock valuations.
	Let $\equiv_\LargestC$ be the relation defined as $\clockval\equiv_\LargestC \clockval'$ iff $\forall \clock\in \Clock$: either $\valuate{\clock}{\clockval}=\valuate{\clock}{\clockval'}$ or ($\valuate{\clock}{\clockval} > \LargestC(\clock)$ and $\valuate{\clock}{\clockval'}>\LargestC(\clock)$).
The relation ${\cal R} = \big\{\big((\loc, \clockval), (\loc, \clockval') \big) \mid \clockval \equiv_\LargestC \clockval' \big\}$ is a bisimulation relation.
\end{lemma}

\begin{example}
Let us recall the motivation for the use of an extrapolation, through the PTA~$\A$ in \cref{figure:bounded_pta:pta}.
After $i$~times through the loop, we get constraints in $\loc_0$ of the form $y-x \leq i$.
The maximal constant\LongVersion{ of the model} is $\maxC(\A)  =5$. 
After five loops, $y$ can be greater than~$5$.
Therefore, we can apply on~$y$ the classical $k$-extrapolation used for TAs (from~\cite{BBLP06}) of the corresponding zone.
More specifically, we consider that when $y > k$, the bounds on~$y$ can be ignored.
The obtained polyhedron is non-convex, but can be split into two convex ones, one where $y \leq k$ (the part without extrapolation) and one with $y > k$ (the part with extrapolation).
This is depicted in \cref{figure:extrapo} where \cref{figure:extrapo:not_extra} is the original clock zone (\LongVersion{formally }$y \leq x + 5 \land x \leq 1 \land x \leq y$) and \cref{figure:extrapo:extra} is its non-convex extrapolation (\LongVersion{formally }$(x \leq y \leq 5 \land x \leq 1) \lor (y \geq 5 \land 0 < x \leq 1)$).
\end{example}
\begin{figure} [tb]
	\centering
	\small
	\begin{subfigure}[b]{0.4\textwidth}
		\centering
		\begin{tikzpicture}[yscale=.5] %
			\tikzstyle{axe} = [line width=1pt, ->, draw=black!80]
			\tikzstyle{fondgris} = [fill=black!5, draw=none]
			\tikzstyle{zone} = [fill=blue!40!white, draw=blue!60!white,line width=1pt]
			\tikzstyle{openzone1} = [fill=blue!20!white, draw=none]
			\tikzstyle{openzone2} = [fill=blue!60!white, draw=none]

			\draw[openzone1] (2, 1) -- (5, 1) -- (5.6, 2)-- (2.6, 2) -- cycle;

			\path[axe] %
				(2, 1) -- ++ (4.5, 0);
			\path[axe] %
				(2, 1) -- ++ (0, 1.5);

			\node at (6.7,1) {$\clocky$};
			\node at (2,2.7) {$\clockx$};

			\foreach \x in {0, 1}
				\draw [-](2, \x+1) -- (1.8, \x+1);
			\foreach \x in {0, 1}
				\node at (1.5,\x+1) {\x};

			\draw [-](2, 1) -- (2, 0.8); 
			\node at (2,0.6) {0};
			\draw [-](2.6, 1) -- (2.6, 0.8); 
			\node at (2.6,0.6) {1};
			\draw [-](5, 1) -- (5, 0.8);
			\node at (5,0.6) {5};
		\end{tikzpicture}
		\caption{A convex clock zone.}
		\label{figure:extrapo:not_extra}
	\end{subfigure}
		\hfill
	\begin{subfigure}[b]{0.5\textwidth}
		\centering
		\begin{tikzpicture}[yscale=.5]  %
			\tikzstyle{axe} = [line width=1pt, ->, draw=black!80]
			\tikzstyle{fondgris} = [fill=black!5, draw=none]
			\tikzstyle{zone} = [fill=blue!40!white, draw=blue!60!white,line width=1pt]
			\tikzstyle{openzone1} = [fill=blue!20!white, draw=none]
			\tikzstyle{openzone2} = [fill=blue!60!white, draw=none]

			\draw[openzone1] (2, 1) -- (5, 1) -- (5, 2)-- (2.6, 2) -- cycle;
			\draw[openzone2] (5, 1.1) -- (6.4, 1.1) -- (6.4, 2) -- (5, 2) -- cycle;

			\path[axe] %
				(2, 1) -- ++ (4.5, 0);
			\path[axe] %
				(2, 1) -- ++ (0, 1.5);

			\node at (6.7,1) {$\clocky$};
			\node at (2,2.7) {$\clockx$};

			\foreach \x in {0, 1}
				\draw [-](2, \x+1) -- (1.8, \x+1);
			\foreach \x in {0, 1}
				\node at (1.5,\x+1) {\x};

			\draw [-](2, 1) -- (2, 0.8); 
			\node at (2,0.6) {0};
			\draw [-](2.6, 1) -- (2.6, 0.8); 
			\node at (2.6,0.6) {1};
			\draw [-](5, 1) -- (5, 0.8);
			\node at (5,0.6) {5};
		\end{tikzpicture}
		\caption{Its non-convex extrapolation.}
		\label{figure:extrapo:extra}
	\end{subfigure}
	\caption{Example illustrating the non-convex parametric extrapolation.}
	\label{figure:extrapo}
\end{figure}
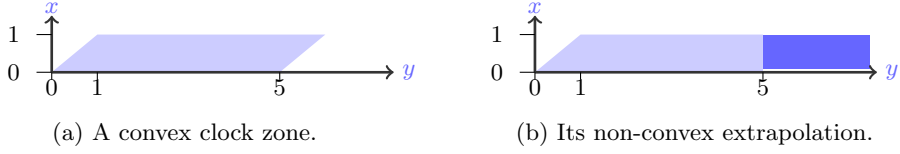

Let us now formally recall from~\cite{ALR15} the concept of $\LargestC$-extrapolation for PTAs.
First, we \LongVersion{need to }recall the \emph{cylindrification} operation, which \LongVersion{is a usual operation that }consists in \emph{unconstraining}\LongVersion{ variable}~$\clock$.

\begin{definition}[Cylindrification \cite{ALR15}]
For a polyhedron~$\C$ and variable~$\clock$, we denote by $\cylinder{\clock}(\C)$ the \emph{cylindrification} of $\C$ along variable $\clock$, \ie{} $\cylinder{\clock}(\C)=\{\clockval \mid \exists \clockval'\in \C, \forall \clock'\neq\clock,\clockval'(\clock')=\clockval(\clock')\text{ and } \clockval(\clock)\geq 0\}$.
\end{definition}

The $(\LargestC,\clock)$-extrapolation is an operation that splits a polyhedron into two polyhedra such that clock~$\clock$ is either less than or equal to~$\LargestC$, or is strictly greater than~$\LargestC$ while being independent from the other variables.

\begin{definition}[$(\LargestC,\clock)$-extrapolation \cite{ALR15}]\label{definition:x-extrapolation}
	Let $\C$ \LongVersion{be }a polyhedron.
	Let $\LargestC \in \setN$\LongVersion{ be a non-negative integer constant} and $\clock$ be a clock.
	The $(\LargestC,\clock)$-extrapolation of~$\C$, denoted by $\Ext{\LargestC}{\clock}(\C)$, is defined as: 
	\[\Ext{\LargestC}{\clock}(\C)= \big(\C\cap (\clock \leq \LargestC)\big) \cup \big( \cylinder{\clock}\big(\C\cap (\clock>\LargestC)\big)\cap (\clock>\LargestC) \big)\text{.}\]
\end{definition}

Given $\symbstate = (\loc, \C)$, we write $\Ext{\LargestC}{\clock}(\symbstate)$ for $\Ext{\LargestC}{\clock}\big(\C\big)$.

We can now consistently define the $\LargestC$-extrapolation operator.

\begin{definition}[$\LargestC$-extrapolation \cite{ALR15}]
    Let $\LargestC \in \setN$\LongVersion{ be a non-negative integer constant} and $\Clock$ be a set of clocks.
    The $(\LargestC,\Clock)$-extrapolation operator $\Ext{\LargestC}{\Clock}$ is defined as the composition (in any order) of all $\Ext{\LargestC}{\clock}$, for all $\clock \in \Clock$.
    When clear from the context we omit $\Clock$ and only write $\LargestC$-extrapolation.%
\end{definition}

\cite[Lemma 1]{ALR15} shows that the order of composition of $(\LargestC,\clock)$-extrapolation does not impact its results, \ie{} $\Ext{\LargestC}{\clock}\big(\Ext{\LargestC}{\clocky}(\C)\big)=\Ext{\LargestC}{\clocky}\big(\Ext{\LargestC}{\clock}(\C)\big)$, 
and \cite[Lemma~5]{ALR15} shows that given a symbolic state $\symbstate$ of a PTA and a non-negative integer $\LargestC$ greater than\LongVersion{ the maximal constant of the PTA}~$\maxC(\A)$, for any clock $\clock$ and parameter valuation $\pval$ such that $(\loc, \clockval)\in \valuate{\Ext{\LargestC}{\clock}(\symbstate)}{\pval}$ is a concrete state, there exists a state $(\loc, \clockval')\in\valuate{\symbstate}{\pval}$ such that $(\loc, \clockval)$ and $(\loc, \clockval')$ are bisimilar.

\subsection{Synthesis with extrapolation}

We now recall the reachability-synthesis algorithm, \LongVersion{that was }formalized in~\cite{JLR15}, and then enhanced with extrapolation (and ``integer hull''---unused here) in~\cite{ALR15}.
We adapt here to our notations a version of reachability-synthesis with the extrapolation\LongVersion{, and write a full proof of correctness (absent from~\cite{ALR15}), also because we will use it and improve it in the remainder of the paper}.

\begin{algorithm}[tb!]
	\SetKwInOut{Input}{input}
	\SetKwInOut{Output}{output}

	\Input{A PTA $\A$, a symbolic state $\symbstate=(\loc, \C)$, a set of target locations $\somelocs$, a set~$\Passed$ of passed states on the current path}
	\Output{Constraint $\K$ over the parameters}

	\LongVersion{\BlankLine}

	\lIf{$\loc \in \somelocs$}{%
		$\K \assign \projectP{\C}$\nllabel{algo:EEF:projection}
	}\Else{%
		$\K \assign \KFalse $\;
        \If{$\symbstate \notin \Passed$}{%
			\For{each outgoing $\edge$ from $\loc$ in $\A$}{
                $\K \assign \K \cup \EEF\big(\A, $\hl{$\Ext{\LargestC}{\Clock}$}$\big(\Succ(\symbstate,\edge)\big), \somelocs, \Passed \cup \{ \symbstate \}\big)$\; \nllabel{algo:EEF:recursion}
			}
		}
	}
	
	\Return{$\K$}
	\caption{$\EEF(\A, \symbstate, \somelocs, \Passed)$}
	\label{algo:EEF}
\end{algorithm}

The goal of \EEF{} given in \cref{algo:EEF} (``\stylealgo{E}'' stands for ``extrapolation'', ``\stylealgo{EF}'' denotes reachability) is to synthesize \LongVersion{parameter }valuation solutions to the reachability-synthesis problem\LongVersion{, \ie{} the valuations for which there exists a run eventually reaching a location in~$\somelocs$}.
\EEF{} proceeds as a post-order traversal of the symbolic reachability tree, and collects all parametric constraints associated with the target locations~$\somelocs$.
In contrast to the classical reachability-synthesis algorithm \EF{} formalized in~\cite{JLR15}, it recursively calls itself (\cref{algo:EEF:recursion}) with the \emph{extrapolation} of the successor of the current symbolic state (this difference is highlighted in yellow in \cref{algo:EEF}).
\newcommand{\commentTheoremEEF}{%
	In order to prove the soundness and completeness of \cref{algo:EEF}, we inductively define, as in~\cite{JLR15}, the \emph{symbolic reachability tree} of~$\A$ as the possibly infinite directed labeled tree $\symtree$ such that:
	\begin{itemize}
		\item the root of $\symtree$ is labeled by the initial symbolic state $\symbstateinit$; 
		\item for every node $n$ of $\symtree$, if $n$ is labeled by some symbolic state $\symbstate$, then for all edges $\edge$ of~$\A$, there exists a child~$n'$ of~$n$ labeled by $\Succ(\symbstate, \edge)$ iff $\Succ(\symbstate, \edge)$ is not empty.
	\end{itemize}

	Algorithm $\EEF$ is a post-order depth-first traversal of some prefix of that tree.
	
	In addition, before we prove \cref{theorem:EEF:soundness}, we need the following lemmas (adapted from~\cite{ALR15}).
	
	We first recall the following lemma (\cite[Lemma~4]{ALR15}):

	\begin{lemma}[{\cite[Lemma~4]{ALR15}}]\label{lemma:ALR15-Lemma4}
		For all parameter valuation~$\pval$, non-negative integer constants $\LargestC$, clock~$\clock$ and valuations set~$\C$, 
		$\valuate{\Ext{\LargestC}{\clock}(\C)}{\pval} = \Ext{\LargestC}{\clock}(\valuate{\C}{\pval})$.
	\end{lemma}

	\begin{lemma}[{\cite[Lemma~5]{ALR15}}]\label{lemma:extbisim:M}
		Let~$\A$ be a PTA and $\symbstate$ be a symbolic state of~$\A$.
		Let $\clock$ be a clock, $\LargestC \in \setN$\LongVersion{ an integer constant} greater than or equal to $\maxC(\A)$, $\pval$ be a parameter valuation and $(\loc, \clockval)\in \valuate{\Ext{\LargestC(\clock)}{\clock}(\symbstate)}{\pval})$ be a concrete state.
		There exists a state $(\loc, \clockval')\in\valuate{\symbstate}{\pval}$ such that $(\loc, \clockval)$ and $(\loc, \clockval')$ are bisimilar.
	\end{lemma}
	
	\begin{remark}
		\cref{lemma:extbisim:Mvect} is the equivalent of \cref{lemma:extbisim:M} for the $\vec{\LargestC}$-extrapolation.
	\end{remark}

	We then prove the following \cref{lemma:extbisimeq}:
	
	\begin{lemma}\label{lemma:extbisimeq}
		Let $\A$ be a PTA.
		For all symbolic states $\symbstate$ and $\symbstate'$, non-negative integer $\LargestC$ greater than\LongVersion{ the maximal constant of the PTA}~$\maxC(\A)$, and parameter valuation $\pval$, such that $\valuate{\Ext{\LargestC}{\Clock}(\symbstate)}{\pval} = \valuate{\Ext{\LargestC}{\Clock}(\symbstate')}{\pval}$, for all states $(\loc, \clockval) \in \valuate{\symbstate}{\pval}$, there exists a state $(\loc, \clockval') \in \valuate{\symbstate'}{\pval}$ such that $(\loc, \clockval)$ and $(\loc, \clockval')$ are bisimilar.
	\end{lemma}

	\begin{proof}
		This is a direct consequence of \cref{lemma:ALR15-Lemma4,lemma:extbisim:M}.
	\end{proof}

}

\LongVersion{\commentTheoremEEF{}}

\cref{algo:EEF} is correct (\ie{} sound and complete):\ShortVersion{\footnote{%
	The proofs of all our results are in \LongVersion{the appendix}\ShortVersion{a technical report~\cite{AA22report}}.
}}

\newcommand{\enonceTheoremEEFcorrectness}{
	Let~$\A$ be a PTA with initial symbolic state~$\symbstateinit$, and $\somelocs \subseteq \Loc$ a set of target locations.
	Assume $\EEF(\A, \symbstateinit, \somelocs, \emptyset)$ terminates.
	We have:
	\begin{enumerate}
		\item Soundness: If $\pval \in \EEF(\A, \symbstateinit, \somelocs, \emptyset)$ then $\somelocs$ is reachable in $\valuate{\A}{\pval}$;
		\item Completeness: For all~$\pval$, if $\somelocs$ is reachable in $\valuate{\A}{\pval}$ then $\pval \in \EEF(\A, \symbstateinit, \somelocs, \emptyset)$.
	\end{enumerate}
}

\begin{theorem}\label{theorem:EEF:soundness}
	\enonceTheoremEEFcorrectness{}
\end{theorem}

\newcommand{\preuveTheoremEEFcorrectness}{%
\begin{proof}
	We reuse here large parts of the proof of \cite[Theorem~2]{ALR15}, as that theorem proves the correctness of a synthesis algorithm using both extrapolation \emph{and} integer hulls---while we use here only extrapolation.
	We give it in full details though, as our formal result will be modified for our subsequent definitions of extrapolations (\eg{} \cref{prop:EEFvect:correctness,prop:EEFhat:correctness}).
\begin{enumerate}
	\item Soundness: this part of the proof is almost exactly the same as in~\cite{JLR15} so we do not repeat it.
	The only difference is that, with the same proof, we actually have a slightly stronger result that holds for any finite prefix of $\symtree$ instead of exactly the one computed by \EF{}: 
		\begin{lemma}\label{lemma:EFtree}
			Let $\treeprefix$ be a finite prefix of
			$\symtree$, on which we apply algorithm \EEF{}.
			Let $n$ be a node of~$\treeprefix$ labeled by some symbolic state $\symbstate$, and such that the subtree rooted at~$n$ has depth~$N$.
			We have: $\pval \in \EEF(\A, \symbstate, \somelocs, \LargestC)$, where $\LargestC$ contains the symbolic states labeling nodes on the path from the root, iff there exists a state $(\loc, \clockval)$ in $\valuate{\symbstate}{\pval}$ and a run $\varrun$ in $\valuate{\A}{\pval}$, with less than $N$ discrete steps, that starts in $(\loc, \clockval)$ and reaches~$\somelocs$.
		\end{lemma}
		Soundness is a direct consequence of \cref{lemma:EFtree}.

	\item Completeness: 
		The proof of this part follows the same general structure as that of \EF{} in~\cite{JLR15} but with additional complications due to the use of the extrapolation.
		We reuse the proof of the result of~\cite{ALR15}, to only cope with extrapolation (without the integer hull defined and used in~\cite{ALR15}).

		Before we start, let us just recall two more results from~\cite{JLR15}:
		\begin{lemma}[{\cite[Lemma~1]{JLR15}}]\label{lemma:next}
			For all parameter valuation~$\pval$, symbolic state~$\symbstate$ and edge~$\edge$, we have $\Succ\big(\valuate{\symbstate}{\pval},\valuate{\edge}{\pval}\big) = \valuate{(\Succ(\symbstate,\edge))}{\pval}$.
		\end{lemma}

		\begin{lemma}[{\cite[Corollary~2]{JLR15}}]\label{cor:next}
			For each parameter valuation $\pval$, reachable symbolic state $\symbstate$, and state $\state$, 
			we have $\state \in \valuate{\symbstate}{\pval}$ if and only if there is a run of $\valuate{\A}{\pval}$ from the initial state leading to~$\state$.
		\end{lemma}

		Now, the algorithm having terminated, it has explored a finite prefix $\treeprefix$ of $\symtree$.
		Let $\pval$ be a parameter valuation.
		Suppose there exists a run $\varrun$ in $\valuate{\A}{\pval}$ that reaches~$\somelocs$.
		Then $\varrun$ is finite and its last state has a location belonging to~$\somelocs$.
		Let $\edge_1,\ldots, \edge_p$ be the edges taken in $\varrun$ and consider the branch in the tree $\treeprefix$ obtained by following this edge sequence on the labels of the arcs in the tree as long as possible.
		If the whole edge sequence is feasible in~$\treeprefix$, then the tree $\treeprefix$ has depth greater than or equal to the size of the sequence and we can apply \cref{lemma:EFtree} to obtain that $\pval \in \EEF(\A, \symbstateinit, \somelocs,\emptyset)$.
		Otherwise, let $\symbstate = (\loc, \C)$ be the symbolic state labeling the last node of the branch, $\edge_k$ be the first edge in $\edge_1,\ldots,\edge_p$ that is not present in the branch and $(\loc, \clockval)$ be the state of $\varrun$ just before taking $\edge_k$.
		Since $(\loc, \clockval)$ has a successor via $\edge_k$, then $\Succ\big( \valuate{\symbstate}{\pval}, \valuate{\edge_k}{\pval} \big)$ is not empty; then using \cref{lemma:next}, $\valuate{\Succ(\symbstate,\edge_k)}{\pval}$ is not empty;
		therefore, $\Succ(\symbstate,\edge_k)$ %
			is not empty.
		Since the node labeled by~$\symbstate$ has no child in~$\treeprefix$, it follows that either $\loc \in \somelocs$ or there exists another node on the branch that is
		labeled by~$\symbstate'$ such that
		$\Ext{\LargestC}{\Clock}(\symbstate) = \Ext{\LargestC}{\Clock}(\symbstate')$.
		
		In the former case, we can apply \cref{lemma:EFtree} to the prefix of~$\varrun$ ending in $(\loc, \clockval)$ and we obtain that $\pval \in \EEF(\A, \symbstateinit, \somelocs, \emptyset)$.  
		
		In the latter case, we have
		$\valuate{\Ext{\LargestC}{\Clock}(\symbstate)}{\pval} = \valuate{\Ext{\LargestC}{\Clock}(\symbstate')}{\pval}$.
		Using now \cref{lemma:extbisimeq}%
		, there exists a state
		$(\loc, \clockval') \in \symbstate'$ that is bisimilar to
		$(\loc, \clockval)$.
		
		Also, by \cref{cor:next}, $(\loc, \clockval')$ is reachable in $\valuate{\A}{\pval}$ via some run $\varrun_1$ along edges $\edge_1\ldots \edge_m$, with $m < k$.
		Also, since $(\loc, \clockval')$ and $(\loc, \clockval)$ are bisimilar, there exists a run $\varrun_2$ that takes the same edges as the suffix of $\varrun$ starting at $(\loc, \clockval)$.
		Let $\varrun'$ be the run obtained by merging $\varrun_1$ and $\varrun_2$ at $(\loc, \clockval')$.
		Run $\varrun'$ has strictly less discrete actions than $\varrun$ and also reaches~$\somelocs$.
		We can thus repeat the same reasoning as we have just done. 
		We can do this only a finite number of times (because the length of the considered run is strictly decreasing) so at some point we have to be in some of the other cases and we
		obtain the expected result. %
	\end{enumerate}
\end{proof}
}

\LongVersion{\preuveTheoremEEFcorrectness{}}

\subsection{Extending the $\LargestC$-extrapolation to individual bounds}

Our first technical contribution is to extend the extrapolation from~\cite{ALR15} to \emph{individual} clock bounds, instead of a global one, in the line of what has been proposed for non-parametric TAs~\cite{BBLP06}.

\begin{definition}[$\vec{\LargestC}$-extrapolation]
\LongVersion{Let $\Clock = \{ \clock_1, \dots, \clock_\ClockCard \} $ the set of clocks of the PTA.}
Let  $\vec{\LargestC} = \{ \LargestC(\clock_1), \dots, \LargestC(\clock_\ClockCard) \}$ be a set of non-negative integer constants.
	The $\vec{\LargestC}$-extrapolation, denoted by $\Ext{\vec{\LargestC}}{\Clock}$, is the composition (in any order) of all $\Ext{\LargestC(\clock)}{\clock}$ for all $\clock \in \Clock$.
\end{definition}

All we need to do for the results from \cite{ALR15} to hold on the $\vec{M}$-extrapolation is to adapt \cite[Lemmas~1 and~5]{ALR15}.

\begin{lemma}\label{lemma:cylcommu}
    For all polyhedra $\C$, integers $\LargestC(\clock), \LargestC(\clock')\geq 0$ and clock variables $\clock$ and $\clock'$, we have $\Ext{\LargestC(\clock)}{\clock}\big(\Ext{\LargestC(\clock')}{\clock'}(\C)\big)=\Ext{\LargestC(\clock')}{\clock'}\big(\Ext{\LargestC(\clock)}{\clock}(\C)\big)$.
    \label{lemma:extxy}
\end{lemma}

\newcommand{\prooflemmacylcommu}{
\begin{proof}
	The result comes from the following facts:
	\begin{enumerate}
		\item $\cylinder{\clock}\big(\cylinder{\clock'}(\C)\big) = \cylinder{\clock'}\big(\cylinder{\clock}(\C)\big)$;
		\item for $\clock \neq \clock', \cylinder{\clock}(\C) \cap (\clock' \bowtie \LargestC(\clock')) = \cylinder{\clock}\big(\C \cap (\clock' \bowtie \LargestC(\clock'))\big)$ for ${\bowtie} \in \{ <,\leq,\geq,> \}$.
	\end{enumerate}
\end{proof}
}

\LongVersion{\prooflemmacylcommu{}}

We now extend \cite[Lemma~5]{ALR15} to $\Ext{\vec{\LargestC}}{}$:

\begin{lemma}[$\vec{\LargestC}$ and bisimilarity]\label{lemma:extbisim:Mvect}
    Let~$\A$ be a PTA and $\symbstate$ be a symbolic state of~$\A$.
    Let $\clock$ be a clock, $\LargestC(\clock) \in \setN$\LongVersion{ an integer constant} greater than or equal to $\maxC^{\clock}(\A)$, $\pval$ be a parameter valuation and $(\loc, \clockval)\in \valuate{\Ext{\LargestC(\clock)}{\clock}(\symbstate)}{\pval})$ be a concrete state.
    There exists a state $(\loc, \clockval')\in\valuate{\symbstate}{\pval}$ such that $(\loc, \clockval)$ and $(\loc, \clockval')$ are bisimilar.
\end{lemma}

\newcommand{\prooflemmaextbisim}{
\begin{proof}
    If $(\loc, \wv{\clockval}{\pval}) \in \symbstate$, then the results holds trivially.
	Otherwise, it means that there exists some clock $\clock$ such that 
    $(\loc, \wv{\clockval}{\pval}) \in \cylinder{\clock}(\symbstate\cap (\clock>\LargestC(\clock)))\cap(\clock>\LargestC(\clock))$. This implies that
		$\valuate{\symbstate (\clock>\LargestC(\clock))}{\pval}\neq \emptyset$ and $\valuate{\clock}{\clockval}>\LargestC(\clock)$.
        Therefore, and using the definition of $\cylinder{\clock}$, there exists $(\loc, \wv{\clockval'}{\pval}) \in \symbstate \cap (\clock>\LargestC(\clock))$ such that for all $\clock'\neq \clock,\valuate{\clock'}{\clockval'}=\valuate{\clock'}{\clockval}$. We also have $\valuate{\clock}{\clockval'}>\LargestC(\clock)$, which means that $\clockval'\equiv_\LargestC \clockval$ and by \cref{lemma:bisim}, we obtain the expected result.
\end{proof}
}

\LongVersion{
	\prooflemmaextbisim{}
}

Given~$\LargestC \in \setN$, given a vector $\vec{\LargestC}$, note that, whenever $\vec{\LargestC}(\clock) \leq \LargestC$ for all $\clock \in \Clock$, then the $\vec{\LargestC}$-extrapolation is necessarily coarser than the $\LargestC$-extrapolation.

Let $\vec{\LargestC}$ be such that, for all~$\clock$, $\vec{\LargestC}(\clock) = \maxC^{\clock}(\A)$.
Let \EEFvect{} denote the modification of \EEF{} where $\Ext{\LargestC}{\Clock}$ is replaced with $\Ext{\vec{\LargestC}}{\Clock}$ (\cref{algo:EEF:recursion} in \cref{algo:EEF}).
That is, instead of computing the $\LargestC$-extrapolation of each symbolic state, we compute its $\vec{\LargestC}$-extrapolation.
\cref{figure:Mext_bounded} illustrates its effect on the state space of \cref{figure:bounded_pta:pta}.

\begin{figure}[t!]
	\centering
	\begin{subfigure}[b]{0.45\textwidth}
		\centering
		\scalebox{.8}{
		\begin{tikzpicture}[>=latex', xscale=1.33, yscale=.7,every node/.style={scale=1}]

		\node[symbstate] at (0,0) (c0) {($\loc_0$, $\clockx = \clocky \leq 1$)};
		\node[symbstate] at (0,-1.5) (c1) {($\loc_0$, $\clockx \leq \clocky \leq \clockx+1 \land \clockx \leq 1$)};
		\node[symbstate] at (2.75,-1.5) (c1') {($\loc_1$, $0 < \paramp \leq 5$)};
		\node[symbstate] at (0,-3) (c2) {($\loc_0$, $\clockx \leq \clocky \leq \clockx+2 \land \clockx \leq 1$)};
		\node[symbstate] at (0,-4.5) (c3) {($\loc_0$, $\clockx \leq \clocky \leq \clockx+3 \land \clockx \leq 1$)};
		\node[symbstate] at (0,-6) (c4) {($\loc_0$, $\clockx \leq \clocky \leq \clockx+4 \land \clockx \leq 1$)};
		\node[symbstate] at (0,-8) (c5) {
			\begin{tabular}{cc}
			& $\clockx \leq \clocky \leq 5 \land \clockx \leq 1$ \\
			$\loc_0$, & $\cup$ \\
			& $5 < \clocky \land 0 < \clockx \leq 1$ \\
			\end{tabular}
		};
		\node[symbstate] at (0,-10.5) (c6) {
			\begin{tabular}{cc}
			& $\clockx \leq \clocky \leq 5 \land \clockx \leq 1$ \\
			$\loc_0$, & $\cup$ \\
			& $5 < \clocky \land \clockx \leq 1$ \\
			\end{tabular}
		};
		\draw[->] (c0) -- (c1);
		\draw[->] (c1.east) -- (c1');
		\draw[->] (c1) -- (c2);
		\draw[->] (c2.east) -- (c1');
		\draw[->] (c2) -- (c3);
		\draw[->] (c3.east) -- (c1');
		\draw[->] (c3) -- (c4);
		\draw[->] (c4.east) -- (c1');
		\draw[->] (c4) -- (c5);
		\draw[->] (c5.east) -- (c1');
		\draw[->] (c5) -- (c6);
		\draw[->] (c6.east) -- (c1'.south);
		\path (c6) edge [max distance=2em, loop left] (c6);
		\end{tikzpicture}
		}
		\caption{Simplified state space of \cref{figure:bounded_pta:pta} with $\LargestC$-extrapolation.}
		\label{figure:Mext_bounded:global}
	\end{subfigure}
	\hspace{10mm}
	\begin{subfigure}[b]{0.45\textwidth}
		\centering
		\scalebox{.8}{
		\begin{tikzpicture}[>=latex', xscale=1.33, yscale=.7,every node/.style={scale=1}]
		\node[symbstate] at (0,0) (c0) {($\loc_0$, $\clockx = \clocky \leq 1$)};
		\node[symbstate] at (0,-2) (c1) {
			\begin{tabular}{cc}
			& $\clockx \leq \clocky \leq 1$ \\
			$\loc_0$, & $\cup$ \\
			& $1 < \clocky \land 0 < \clockx \leq 1$ \\
			\end{tabular}
		};
		\node[symbstate] at (2.5,-2) (c1') {($\loc_1$, $0 < \paramp \leq 5$)};
		\node[symbstate] at (0,-4.5) (c2) {
			\begin{tabular}{cc}
			& $\clockx \leq \clocky \leq 1$ \\
			$\loc_0$, & $\cup$ \\
			& $1 < \clocky \land \clockx \leq 1$ \\
			\end{tabular}
		};
		\draw[->] (c0) -- (c1);
		\draw[->] (c1) -- (c1');
		\draw[->] (c1) -- (c2);
		\draw[->] (c2.east) -- (c1');
		\path (c2) edge [max distance=2em, loop left] (c2);
		\end{tikzpicture}
		}
		\caption{Simplified state space of \cref{figure:bounded_pta:pta} with $\vec{\LargestC}$-extrapolation where $\LargestC(\clockx)=5$ and $\LargestC(\clocky)=1$.
As $\vec{\LargestC}$-extrapolation differentiates the maximal constant of each clock, the extrapolation is applied on $\clocky$ after only one loop.}
		\label{figure:Mext_bounded:vecteur}
	\end{subfigure}
	\caption{Comparison between ${\LargestC}$-extrapolation and $\vec{\LargestC}$-extrapolation.}
	\label{figure:Mext_bounded}
\end{figure}
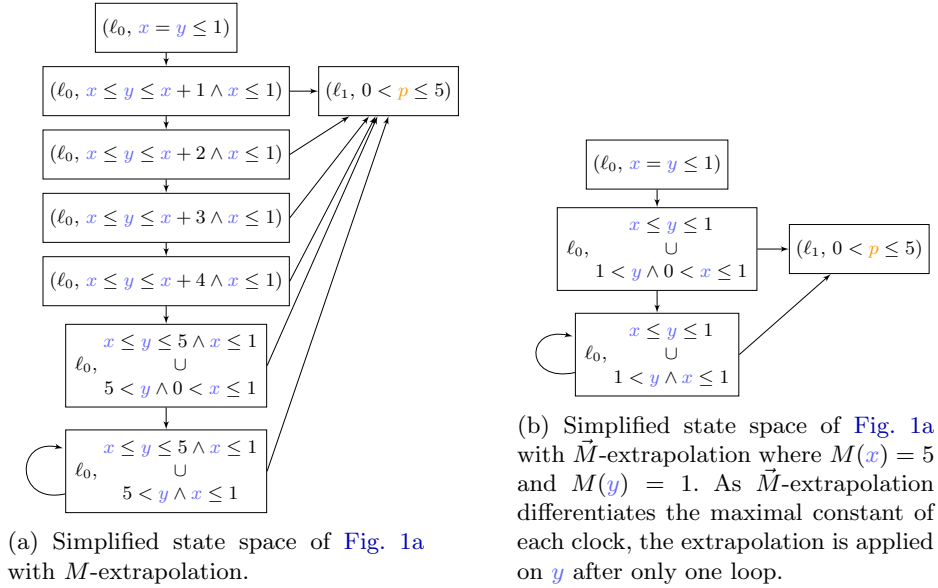

\newcommand{\enonceTheoremEEFvectCorrectness}{
	Let~$\A$ be a PTA with initial symbolic state~$\symbstateinit$, and $\somelocs \subseteq \Loc$ a set of target locations.
	Assume $\EEFvect(\A, \symbstateinit, \somelocs, \emptyset)$ terminates.
	We have:
	\begin{enumerate}
		\item Soundness: If $\pval \in \EEFvect(\A, \symbstateinit, \somelocs, \emptyset)$ then $\somelocs$ is reachable in $\valuate{\A}{\pval}$;
		\item Completeness: For all~$\pval$, if $\somelocs$ is reachable in $\valuate{\A}{\pval}$ then $\pval \in \EEFvect(\A, \symbstateinit, \somelocs, \emptyset)$.
	\end{enumerate}
}

\begin{proposition}\label{prop:EEFvect:correctness}
	\enonceTheoremEEFvectCorrectness{}
\end{proposition}

\newcommand{\preuveTheoremEEFvectCorrectness}{%
\begin{proof}
	The result follows immediately from the proof of \cref{theorem:EEF:soundness}, by applying \cref{lemma:extbisim:Mvect} instead of \cref{lemma:extbisimeq}.
\end{proof}
}

\LongVersion{\preuveTheoremEEFvectCorrectness{}}

\section{$\vec{\LargestC}$-extrapolation on unbounded PTAs} \label{section:M-Ex_unbounded}

In this section, we extend the $\vec{\LargestC}$-extrapolation to subclasses of (unbounded) PTAs.
This requires to be able to identify for each clock $\clock \in \Clock$ a constant $\LargestC(\clock)$ such that given a symbolic state $\symbstate$ and a parameter valuation $\pval$, for any concrete state in $\pval(\Ext{\LargestC(\clock)}{\clock}(\symbstate))$ there exists a bisimilar state in $\pval(\symbstate)$, \ie{} \cref{lemma:extbisim:Mvect} holds\LongVersion{ true}.
\LongVersion{%

}%
We will consider
\begin{ienumerate}%
	\item L-PTAs and U-PTAs (\cref{ss:L-U}),
	\item bounded PTAs with additional \emph{unbounded} lower-bound or upper-bound parameters (\cref{ss:bPTA+L-U}),
		and
	\item the full class of PTAs to which we apply extrapolation only on bounded parameters (\cref{ss:partial-PTAs}).
\end{ienumerate}

\subsection{$\vec{\LargestC}$-extrapolation on unbounded L-PTAs and U-PTAs}\label{ss:L-U}

\LongVersion{%
\subsubsection{Recalling L-PTAs and U-PTAs}
}

\LongVersion{%
	We will use results from \cite{BlT09}, where the authors propose a constant~$\LargestP$ for unbounded parameters such that any parameter valuation greater than~$\LargestP$ will exhibit similar behaviors in regard of infinite accepting runs.
	Specifically, a (different) constant~$\LargestP$ can be computed on unbounded L-PTAs and U-PTAs\LongVersion{, which are subset of the general~PTAs}.
}

\LongVersion{%
	First, let us recall the definitions of L-PTAs and U-PTAs~\cite{BlT09}.
	An L-PTA (respectively U-PTA) is a PTA where each parameter always appears as a lower- (respectively upper-)bound when compared to a clock.
}

\begin{definition}[L-PTA and U-PTA \cite{BlT09}]
	A PTA $\A$ is an \emph{L-PTA} (resp.\ \emph{U-PTA}) if, for each guard $\clock \compOp \sum_{1 \leq i \leq \ParamCard} \alpha_i \param_i + z$ of~$\SCGuards(\A)$, for all~$i$\ShortVersion{ with $\alpha_i \neq 0$}:
	\begin{itemize}
		\LongVersion{\item $\alpha_i = 0$, or}
		\item $\alpha_i > 0$ and ${\compOp} \in \{\geq, >\}$ (respectively ${\compOp} \in \{<, \leq \}$), or
		\item $\alpha_i < 0$ and ${\compOp} \in \{<, \leq \}$ (respectively ${\compOp} \in \{\geq, >\}$).
	\end{itemize}
\end{definition}

L-PTAs and U-PTAs feature a well-known monotonicity property: enlarging a parameter valuation in a U-PTA (resp.\ decreasing in an L-PTA) can only \emph{add} behaviors, as recalled in the following lemma:

\begin{lemma}[\cite{BlT09}]\label{lemma:monotonicity}
	Given a U-PTA (resp.\ L-PTA)~$\A$,
	given two valuations $\pval_1, \pval_2$ with
	$\pval_1 \leq \pval_2$ (resp.\ $\pval_1 \geq \pval_2$),
	then
	$\Runs(\pval_1(\A)) \subseteq \Runs(\pval_2(\A))$.
\end{lemma}

For any L-PTA $\A$, as per \cite[Theorem~3]{BlT09}, there exists a constant bound~$\LargestP$, such that for all valuations $\pval_1, \pval_2$ with $\pval_1 \geq \pval_2 \geq \pval_{\LargestP}$ (where $\pval_{\LargestP}$ denotes the parameter valuation assigning $\LargestP$ to each parameter), if $\pval_2(\A)$ provides an infinite accepting run then so does $\pval_1(\A)$.
\LongVersion{%
	Since $\A$ is an L-PTA, $\pval_2(\A)$ includes all the possible executions of $\pval_1(\A)$, which is given by \cref{lemma:monotonicity}.
	That is,
if $\pval_1(\A)$ yields an infinite accepting run, then so does $\pval_2(\A)$.
Therefore, for any valuations $\pval \geq \pval_{\LargestP}$ and $\pval' \geq \pval_{\LargestP}$, $\pval(\A)$ yields an infinite accepting run iff $\pval'(\A)$ yields an infinite accepting run.

}%
A dual result is shown for U-PTAs\LongVersion{ in~\cite[Theorem~6]{BlT09}}.
\LongVersion{%
	For any U-PTA $\A$, there exists a constant bound $\LargestP$ such that for all valuations $\pval_1, \pval_2$ with $\pval_1 \geq \pval_2 \geq \pval_{\LargestP}$, if $\pval_1(\A)$ yields an infinite accepting run then so does $\pval_2(\A)$.
	As $\A$ is a U-PTA, $\pval_1(\A)$ includes all the possible executions of $\pval_2(\A)$, hence if $\pval_2(\A)$ yields an infinite accepting run then so does $\pval_1(\A)$.
	Therefore, for a given valuation $\pval \geq \pval_{\LargestP}$, if $\pval(\A)$ yields an infinite accepting run, then so does $\pval '(\A)$ for any $\pval ' \geq \pval_{\LargestP}$.
}%
Formally:

\begin{lemma}[{\cite[Theorems~3 and~6]{BlT09}}]\label{lemme-BlT09-infinite}
	Given a U-PTA (resp.\ L-PTA)~$\A$ with $\LargestP$ the constant bound defined in~\cite{BlT09},
	given two valuations $\pval_1 \geq \pval_{\LargestP}$ and $\pval_2 \geq \pval_{\LargestP}$,
	there exists an infinite accepting run in $\valuate{\A}{\pval_1}$
	iff
	there exists an infinite accepting run in $\valuate{\A}{\pval_2}$.
\end{lemma}

\paragraph{Computation of $\overLargestP$}
Given an L-PTA (respectively U-PTA) $\A$, the value given in~\cite{BlT09} is $\LargestP = k (R + 1) + c + 1$ (respectively $\LargestP = 8 k (R + 1) + c + 1$), where $k$ is the number of parametric clocks of~$\A$, $R$ is the number of clock regions obtained when the parameter valuation is $0$ for all parameters, and $c$ is the greatest non-parametric constant in absolute value among all linear expressions.
More precisely, all linear expression being of the form $\sum_{1 \leq i \leq \ClockCard} \alpha_i \clock_i + \sum_{1 \leq j \leq \ParamCard} \beta_j \param_j + d \compOp 0$, $c$ is the maximum over all $|d|$.
Although $k$ and $c$ are obtained syntactically, $R$ needs to be computed.
As $\LargestP$ acts as a lower bound, using an over-approximation of~$R$ would still guarantee the correctness of \cref{lemme-BlT09-infinite}.
From \cite[Lemma 4.5]{AD94}, the number of clock regions is bounded by $\widehat{R} = 2^{|\Clock|}  |\Clock|!  \prod_{\clock \in \Clock} (2c_\clock + 2)$ with $\Clock$ the set of clocks and $c_\clock$ the greatest constant over~$\clock$ (either as a upper or lower bound)---which can both be obtained syntactically.
We define $\overLargestP$ as the constant defined in~\cite{BlT09} for an L-PTA (resp.\ U-PTA)~$\A$, where we use~$\widehat{R}$ \LongVersion{(the aforementioned over-approximation of the number of clock regions) }instead of\LongVersion{ their actual number}~$R$.

\subsubsection{Formal results}

We first adapt \cref{lemme-BlT09-infinite} to our new constant $\overLargestP$:

\begin{lemma}\label{lemme-BlT09-infinite-adaptation}
	Given a U-PTA (resp.\ L-PTA)~$\A$,
	given two valuations $\pval_1 \geq \pval_{{\overLargestP}}$ and $\pval_2 \geq \pval_{{\overLargestP}}$,
	there exists an infinite accepting run in $\valuate{\A}{\pval_1}$
	iff
	there exists an infinite accepting run in $\valuate{\A}{\pval_2}$.
\end{lemma}

\newcommand{\prooflemmeBlTinfiniteadaptation}{%
\begin{proof}
	From the fact that we use in the computation of ${\overLargestP}$ an over-approximation on the number of clock regions (with $R \leq \widehat{R}$), giving 
	$\LargestP \leq {\overLargestP}$.
\end{proof}
}

\prooflemmeBlTinfiniteadaptation{}

We can now prove the correctness of extrapolation for unbounded L-PTAs and U-PTAs.
\LongVersion{

}%
Let $\widehat{\LargestC} = \{ \LargestC(\clock_1), \dots, \LargestC(\clock_\ClockCard) \}$ such that $\LargestC(\clock_i)$ is the maximal constant of clock~$\clock_i$ when bounding all unbounded parameters with~$\overLargestP$.
Let \EEFhat{} denote the modification of \EEF{} where $\Ext{\LargestC}{\Clock}$ is replaced with $\Ext{\widehat{\LargestC}}{\Clock}$\LongVersion{ (\cref{algo:EEF:recursion} in \cref{algo:EEF})}.
\LongVersion{%
	That is, instead of computing the $\Ext{\LargestC}{}$-extrapolation of each symbolic state, we compute its $\Ext{\widehat{\LargestC}}{}$-extrapolation.
}%

\begin{example}

\cref{figure:unbounded_pta} illustrates the effects of the $\widehat{\LargestC}$-extrapolation on the unbounded U-PTA of \cref{figure:unbounded_pta:pta}.
\cref{figure:unbounded_pta:ss} displays its (simplified) infinite state space.
The valuation of\LongVersion{ parameter}~$\param$ can be any value in~$\setQplus$.
\cref{figure:unbounded_pta:Mext} shows the state space obtained with the $\widehat{\LargestC}$-extrapolation.
Note that the state space is now finite.

\begin{figure} [tb]
	\begin{subfigure}[b]{0.3\textwidth}
		\centering
		\begin{tikzpicture}[PTA, thin, scale=.6]

			\node[location, initial] at(0,0) (l0) {$\loc_0$};
			\node[location] at(3,0) (l1) {$\loc_1$};
			\node [invariant, below] at (0,-0.5) {$\clockx \leq 1$};
			\node [invariant, below] at (3,-0.5) {$\clocky \leq \paramp$};

			\path
				(l0) edge[loop] node [above,align=center] {$\clockx = 1; \clockx \assign 0$} (l0)
				(l0) edge[] node {$ \clocky \geq 1$}  (l1)
			;
		\end{tikzpicture}
		\caption{Unbounded U-PTA}
		\label{figure:unbounded_pta:pta}
	\end{subfigure}
	\hspace{1em}
	\begin{subfigure}[b]{0.3\textwidth}
	\centering
	\small
		\begin{tikzpicture}[PTA, thin]

			\node[location, initial] at(0,0) (l0) {$\loc_0$};
			\node[location] at(3,0) (l1) {$\loc_1$};
			\node [invariant, below] at (0,-0.5) {$\clockx \leq 1$};

			\path
				(l0) edge[loop] node [above] {$\clockx \assign 0$} (l0)
				(l0) edge[] node {$1< \clocky \land \clockx = \paramp$}  (l1)
			;
		\end{tikzpicture}
	\caption{Unbounded PTA}
	\label{figure:partial_ext}
	\end{subfigure}
	\hspace{1em}
	\begin{subfigure}[b]{0.3\textwidth}
	\centering
	\small
		\begin{tikzpicture}[PTA, thin]

			\node[location, initial] at(0,0) (l0) {$\loc_0$};
			\node[location] at(2,0) (l1) {$\loc_1$};

			\node [invariant, below, yshift=.7em] at (0,-0.5) {$\clockx \leq 1 \land \clocky \leq \paramp$};

			\path
				(l0) edge[loop] node [above, align=center] {$\clockx = 1  ; \clockx \assign 0$} node[below]{$\styleact{a}$} (l0)
			;
			\path
				(l0) edge node [above,align=center] {$\styleact{b}$ \\ $\clocky \leq \paramp$} (l1)
			;
		\end{tikzpicture}
		\caption{Unbounded U-PTA}
		\label{figure:PTA:IM}
	\end{subfigure}
	\caption{Three toy PTAs}
\end{figure}
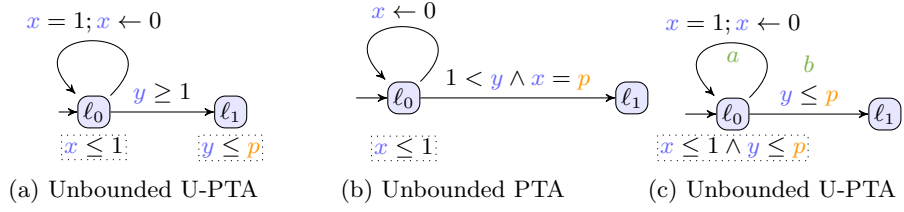

\begin{figure} [t]
	\small
	\begin{subfigure}[t]{0.38\textwidth}
		\centering
		\raisebox{1.55cm}{
		\scalebox{.8}{
		\begin{tikzpicture}[>=latex', xscale=1.3, yscale=.7,every node/.style={scale=1}]
		\node[symbstate] at (0,0) (c0) {($\loc_0$, $\clockx=\clocky \leq 1$)};
		\node[symbstate] at (2.5,0) (c0') {($\loc_1$, $\paramp \geq 1$)};
		\node[symbstate] at (0,-2) (c1) {($\loc_0$, $\clocky=\clockx+1 \land \clockx \leq 1$)};
		\node[symbstate] at (2.5,-2) (c1') {($\loc_1$, $\paramp \geq 1$)};
		\node[infinitesymbstate] at (0,-4) (ci) {($\loc_0$, $\clocky=\clockx+i \land \clockx \leq 1$)};
		\node[infinitesymbstate] at (2.5,-4) (ci') {($\loc_1$, $\paramp \geq i$)};
		\draw[->] (c0) -- (c0');
		\draw[->] (c0) -- (c1);
		\draw[->] (c1) -- (c1');
		\draw[->, blue,dashed] (c1) -- (ci);
		\draw[->] (ci) -- (ci');
		\draw[->, blue,dashed] (ci) -- (0,-5);
		\end{tikzpicture}
		}
		}
		\vspace*{-1.8cm}
		\caption{Simplified state space of \cref{figure:unbounded_pta:pta}.}
		\label{figure:unbounded_pta:ss}
	\end{subfigure}
	\hspace{-2.3cm}
	\begin{subfigure}[t]{0.6\textwidth}
		\centering
		\captionsetup{width=1.55\linewidth}
		\scalebox{.8}{
		\begin{tikzpicture}[>=latex', xscale=1.33, yscale=.66,every node/.style={scale=1}]
		\node[symbstate] at (0,-0.5) (c0) {($\loc_0$, $\clockx=\clocky \leq 1$)};
		\node[symbstate] at (2.5,-0.5) (c0') {($\loc_1$, $\paramp \geq 1$)};
		\node[symbstate] at (0,-2.2) (c1) {($\loc_0$, $\clocky=\clockx+1 \land \clockx \leq 1$)};
		\node[symbstate] at (2.5,-2.2) (c1') {($\loc_1$, $\paramp \geq 1$)};
		\node[symbstate] at (0,-4) (c1033) {($\loc_0$, $\clocky=\clockx+1033 \land \clockx \leq 1$)};
		\node[symbstate] at (2.5,-4) (c1033') {($\loc_1$, $\paramp \geq 1033$)};
		\node[symbstate] at (0,-6) (c1034) {
			\begin{tabular}{cc}
			& $\clocky=1034 \land \clockx=0$ \\
			$\loc_0$, & $\cup$ \\
			& $\clocky > 1034 \land 0 < \clockx \leq 1$ \\
			\end{tabular}
		};
		\node[symbstate] at (2.5,-6) (c1034') {$\loc_1$, $\paramp \geq 1034$};
		\node[symbstate] at (0,-8) (c1035) {$\loc_0$, $\clocky > 1034 \land \clockx \leq 1$};
		\node[symbstate] at (2.5,-8) (c1035') {$\loc_1$, $\paramp > 1034$};
		\draw[->] (c0) -- (c0');
		\draw[->] (c0) -- (c1);
		\draw[->] (c1) -- (c1');
		\draw[->,dashed] (c1) -- (c1033);
		\draw[->] (c1033) -- (c1033');
		\draw[->] (c1033) -- (c1034);
		\draw[->] (c1034) -- (c1034');
		\draw[->] (c1034) -- (c1035);
		\draw[->] (c1035) -- (c1035');
		\path (c1035) edge [max distance=2em, loop left] (c1035);
		\end{tikzpicture}
		}
		\hspace{-6cm}
		\caption{Simplified state space of \cref{figure:unbounded_pta:pta} with the $\widehat{\LargestC}$-extrapolation where $\LargestC(\clockx)=1$ and $\LargestC(\clocky)=1034$, computed using $\overLargestP$. The dashed link represents a succession of 1031 intermediate states where the value of $\clocky$ grows from $\clockx+1$ to $\clockx+1033$. }
		\label{figure:unbounded_pta:Mext}
	\end{subfigure}
	\caption{Example of an unbounded PTA generating an infinite state space.
	}
	\label{figure:unbounded_pta}
\end{figure}
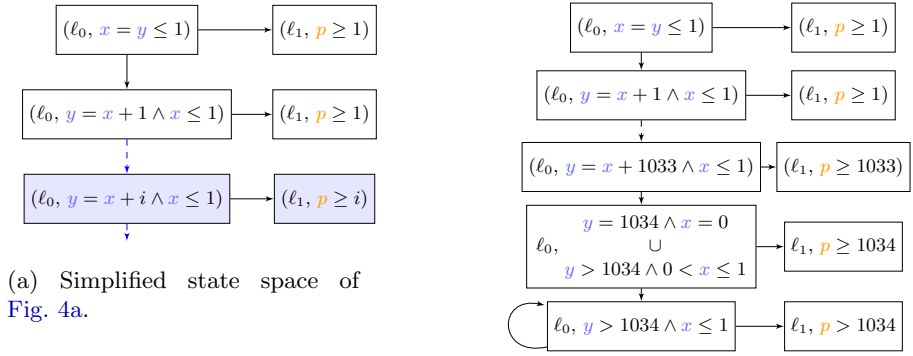

\end{example}

\newcommand{\enonceTheoremEEFhatCorrectness}{
	Let~$\A$ be an L-PTA or U-PTA with initial\LongVersion{ symbolic} state~$\symbstateinit$, and $\somelocs \subseteq \Loc$ a set of target locations.
	Assume $\EEFhat(\A, \symbstateinit, \somelocs, \emptyset)$ terminates.
	We have:
	\begin{enumerate}
		\item Soundness: If $\pval \in \EEFhat(\A, \symbstateinit, \somelocs, \emptyset)$ then $\somelocs$ is reachable in $\valuate{\A}{\pval}$;
		\item Completeness: For all~$\pval$, if $\somelocs$ is reachable in $\valuate{\A}{\pval}$ then $\pval \in \EEFhat(\A, \symbstateinit, \somelocs, \emptyset)$.
	\end{enumerate}
}

\begin{proposition}\label{prop:EEFhat:correctness}
	\enonceTheoremEEFhatCorrectness{}
\end{proposition}

\newcommand{\preuveTheoremEEFhatCorrectness}{%

We first prove the following lemma, which adapts \cref{lemma:bisim} to L-PTAs and U-PTAs.

\begin{lemma}\label{lemma:bisim:Mhat}
	Let~$\A$ be an L-PTA or a U-PTA. 
	Given clock~$\clock$, let $\LargestC(\clock)$ be an integer constant greater than or equal to  the maximal constant $\maxC^{\clock}(\A)$ of clock~$\clock$ when bounding all unbounded parameters with $\overLargestP$.
	For a given parameter valuation $\pval(\A)$ of $\A$, let $\clockval, \clockval'$ be two clock valuations.
	Let $\equiv_\LargestC$ be the relation defined as $\clockval\equiv_\LargestC \clockval'$ iff $\forall \clock\in \Clock$: either $\valuate{\clock}{\clockval}=\valuate{\clock}{\clockval'}$ or ($\valuate{\clock}{\clockval} > \LargestC(\clock)$ and $\valuate{\clock}{\clockval'}>\LargestC(\clock)$).
The relation ${\cal \widehat{R}} = \big\{\big((\loc, \clockval), (\loc, \clockval') \big) \mid \clockval \equiv_\LargestC \clockval' \big\}$ is a bisimulation relation.
\end{lemma}

\begin{proof}
Any valuation $\valuate{\clock}{\clockval} > \LargestC(\clock)$ implies a parameter valuation $\pval$ greater than or equal to~$\pval_{\overLargestP}$.
And we know by \cref{lemme-BlT09-infinite-adaptation} that either for all valuation $\pval \geq \pval_{\overLargestP}$, $\pval (\A)$ accepts an infinite accepting run, or for all valuation $\pval \geq \pval_{\overLargestP}$, $\pval (\A)$ does not accept an infinite accepting run.
As checking infinite accepting run can be used to reachability (for instance, by introducing an unguarded self-loop on each location matching the accepting condition), this implies that any reachable location can be reached with a clock valuation $\clockval$ such that for any $\clock_i$, $\valuate{\clock_i}{\clockval} \leq \LargestC(\clock_i)$.
As a result, relation ${\cal \widehat{R}}$ preserve the bisimilarity of relation ${\cal R}$ from \cref{lemma:bisim}.

\end{proof}

We then prove the following lemma, equivalent to \cref{lemma:extbisim:Mvect}.

\begin{lemma}[$\widehat{\LargestC}$ and bisimilarity]\label{lemma:extbisim:Mhat} 
    Let~$\A$ be an L-PTA or a U-PTA and $\symbstate$ be a symbolic state of~$\A$.

    Let $\clock$ be a clock, $\LargestC(\clock) \in \setN$\LongVersion{ an integer constant} greater than or equal to %
 the maximal constant $\maxC^{\clock}(\A)$ of clock~$\clock$ when bounding all unbounded parameters with~$\overLargestP$, $\pval$ be a parameter valuation and $(\loc, \clockval)\in \valuate{\Ext{\LargestC(\clock)}{\clock}(\symbstate)}{\pval})$ be a concrete state.
There exists a state $(\loc, \clockval')\in\valuate{\symbstate}{\pval}$ such that $(\loc, \clockval)$ and $(\loc, \clockval')$ are bisimilar.
\end{lemma}

\begin{proof}
The result follows immediately from the proof of \cref{lemma:extbisim:Mvect}, by applying \cref{lemma:bisim:Mhat} instead of \cref{lemma:bisim}.
\end{proof}

We can proceed with the proof of \cref{prop:EEFhat:correctness}:

\begin{proof}
	The result follows immediately from the proof of \cref{theorem:EEF:soundness}, by applying \cref{lemma:extbisim:Mhat} instead of \cref{lemma:extbisimeq}.
\end{proof}
}

\LongVersion{\preuveTheoremEEFhatCorrectness{}}

\subsection{$\vec{\LargestC}$-extrapolation on PTAs with unbounded lower or upper bound parameters}\label{ss:bPTA+L-U}

\ShortVersion{A subclass of PTAs can be turned into L-PTAs or U-PTAs (only) for the sake of computing $\overLargestP$.}%
\LongVersion{The method described previously can be adapted to a subclass of PTAs that can be turned into L-PTAs or U-PTAs (only) for the sake of computing the constant bound~$\overLargestP$.}
Let us first define this subclass:

\begin{definition}[bPTA+L and bPTA+U]\label{definition:bPTA+L-U}
	Let $\A$ be a PTA.
	$\A$ is a \emph{bounded PTA with unbounded lower-(resp.\ upper-)bound parameters}, or bPTA+L (resp.\ bPTA+U), if for each guard $\clock \compOp \sum_{1 \leq i \leq \ParamCard} \alpha_i \param_i + z$ of~$\SCGuards(\A)$, for all~$i$:
	\begin{itemize}
		\item $\PDomain(\param_i) \in \setQ \times \setQ$ (\ie{} $\param_i$ is bounded), or
		\LongVersion{\item }$\alpha_i = 0$, or
		\item $\alpha_i > 0$ and ${\compOp} \in \{\geq, >\}$ (respectively ${\compOp} \in \{<, \leq \}$), or
		\item $\alpha_i < 0$ and ${\compOp} \in \{<, \leq \}$ (respectively ${\compOp} \in \{\geq, >\}$).
	\end{itemize}
\end{definition}

Let $\A$ be a bPTA+L (resp.~bPTA+U).
We denote by $\valuateLU{\A}$ the L-PTA (resp.\ U-PTA) obtained from~$\A$ by valuating the bounded parameters as follows: we replace each bounded parameter~$\param_i$ within a guard or invariant with its lower bound $\PDomain^-(\param_i)$ if it appears negatively ($\alpha_i < 0$) or with its upper bound $\PDomain^+(\param_i)$ otherwise.
Clearly, if $\A$ is a bPTA+L (resp.~bPTA+U) then $\valuateLU{\A}$ is an L-PTA (resp.~U-PTA). %

We first valuate bounded parameters to turn a bPTA+L (resp.~bPTA+U) $\A$ into an L-PTA (resp.~U-PTA).
This is obtained by transforming~$\A$ such that, in every guard and invariant, any bounded parameter of positive coefficient $\alpha_i$ is replaced with its upper bound and any bounded parameter of negative coefficient $\alpha_i$ with its lower bound.

\begin{definition}[Bounded valuation of a bPTA+L or bPTA+U]\label{definition:bPTA+L-U:transformation}
	Let $\A$ be a bPTA+L (resp.~bPTA+U).
	Let $\valuateLU{\A}$ be the modification of~$\A$ where
		for each guard $\clock \compOp \sum_{1 \leq i \leq \ParamCard} \alpha_i \param_i + z \in \SCGuards(\A)$,
		for each bounded $\param_i \in \Param$,
		\begin{ienumerate}
			\item if $\alpha_i < 0$, $\param_i$ is replaced by $\PDomain^-(\param_i)$,
			\item if $\alpha_i > 0$, $\param_i$ is replaced by $\PDomain^+(\param_i)$, and
			\item $\param_i$ is replaced with~0 otherwise.
		\end{ienumerate}
\end{definition}

\begin{example}
To illustrate \cref{definition:bPTA+L-U} we modify \cref{figure:unbounded_pta:pta} by adding a bounded parameter.
\cref{figure:bounded_val:pta} is a bPTA+U $\A$ with $\styleparam{q}$ bounded between~$1$ and~$2$, and $\paramp$ unbounded.
\cref{figure:bounded_val:bounded_val} is the bounded valuation $\A'$ of $\A$, as defined in \cref{definition:bPTA+L-U:transformation}.
Note that in this example $\A'$ does not describe a behavior that belongs to $\A$, as parameter $\styleparam{q}$ is valuated to $1$ in the guard where it occurs with a negative sign, while it is valuated to $2$ in the guard where it occurs with a positive sign.
It will nevertheless be useful to determine a constant bound for~$\A$.

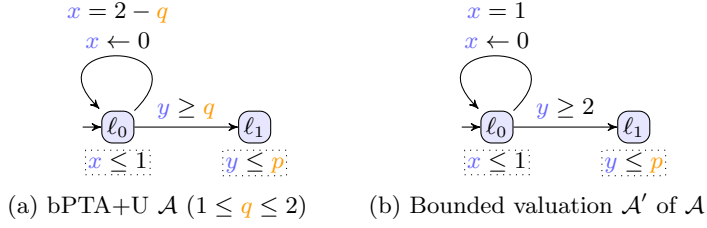
\begin{figure} [tb]
	\centering
	\small
	\begin{subfigure}[b]{0.4\textwidth}
		\centering
		\begin{tikzpicture}[PTA, thin, scale=.6]

			\node[location, initial] at(0,0) (l0) {$\loc_0$};
			\node[location] at(3,0) (l1) {$\loc_1$};
			\node [invariant, below] at (0,-0.5) {$\clockx \leq 1$};
			\node [invariant, below] at (3,-0.5) {$\clocky \leq \paramp$};

			\path
				(l0) edge[loop] node [above,align=center] {$\clockx = 2 - \styleparam{q}$ \\$ \clockx \assign 0$} (l0)
				(l0) edge[] node {$ \clocky \geq \styleparam{q}$}  (l1)
			;
		\end{tikzpicture}
		\caption{bPTA+U $\A$ ($1 \leq \styleparam{q} \leq 2$)}
		\label{figure:bounded_val:pta}
	\end{subfigure}
	\begin{subfigure}[b]{0.4\textwidth}
		\centering
		\begin{tikzpicture}[PTA, thin, scale=.6]

			\node[location, initial] at(0,0) (l0) {$\loc_0$};
			\node[location] at(3,0) (l1) {$\loc_1$};
			\node [invariant, below] at (0,-0.5) {$\clockx \leq 1$};
			\node [invariant, below] at (3,-0.5) {$\clocky \leq \paramp$};

			\path
				(l0) edge[loop] node [above,align=center] {$\clockx = 1$ \\$ \clockx \assign 0$} (l0)
				(l0) edge[] node {$ \clocky \geq 2$}  (l1)
			;
		\end{tikzpicture}
	\caption{Bounded valuation $\A'$ of $\A$}
	\label{figure:bounded_val:bounded_val}
	\end{subfigure}
	\caption{A bPTA+U and its bounded valuation.}
\end{figure}

\end{example}

\subsubsection{Correctness of the transformation}

Trivially, we get that the PTA $\valuateLU{\A}$ is an L-PTA (or U-PTA).

\begin{lemma}\label{lemma:bPTA+L-U=L-U}
	Let $\A$ be a bPTA+L (resp.~bPTA+U).
	Then $\valuateLU{\A}$ is an L-PTA (resp.~U-PTA).
\end{lemma}
\begin{proof}
	Assume $\A$ is a bPTA+L (resp.~bPTA+U).
	When building $\valuateLU{\A}$, any occurrence of a bounded parameter is replaced by its constant bounds.
	In addition, all unbounded parameters from~$\A$ are, by \cref{definition:bPTA+L-U}, lower-bound (resp.\ upper-bound) parameters.
	Therefore, the only remaining parameters in $\valuateLU{\A}$ are lower-bound (resp.\ upper-bound) parameters.
	Therefore, $\valuateLU{\A}$ is an L-PTA (resp.~U-PTA).
\end{proof}
\subsubsection{Method}
Our method is then as follows: given a bPTA+L (resp.~bPTA+U)~$\A$,
\begin{ienumerate}
	\item we construct the L-PTA (resp.\ U-PTA) $\valuateLU{\A}$, and
	\item we then compute the bound $\overLargestP$ on the obtained L-PTA (resp.\ U-PTA) $\valuateLU{\A}$ (using the technique given in \cref{ss:L-U}).
\end{ienumerate}%
Let $\constantBoundLUplus$ denote this result.

Let $\overline{\LargestC} = \{ \LargestC(\clock_1), \dots, \LargestC(\clock_\ClockCard) \}$ such that $\LargestC(\clock_i)$ is the maximal constant of clock~$\clock_i$ when bounding in~$\A$ all unbounded parameters with~$\constantBoundLUplus$.
Let \EEFbar{} denote the modification of \EEF{} where $\Ext{\LargestC}{\Clock}$ is replaced with $\Ext{\overline{\LargestC}}{\Clock}$\LongVersion{ (\cref{algo:EEF:recursion} in \cref{algo:EEF})}.
\LongVersion{%
	That is, instead of computing the $\Ext{\LargestC}{}$-extrapolation of each symbolic state, we compute its $\Ext{\overline{\LargestC}}{}$-extrapolation, where $\overline{\LargestC}$ was obtained using the $\constantBoundLUplus$ computed on the L-PTA (or U-PTA) when valuating the bounded parameters with their bounds.
}%

\LongVersion{
\subsubsection{Correctness}
}
\newcommand{\enonceTheoremEEFbarCorrectness}{
	Let~$\A$ be a bPTA+L or bPTA+U with initial\LongVersion{ symbolic} state~$\symbstateinit$, and $\somelocs \subseteq \Loc$ a set of target locations.
	Assume $\EEFbar(\A, \symbstateinit, \somelocs, \emptyset)$ terminates.
	We have:
	\begin{enumerate}
		\item Soundness: If $\pval \in \EEFbar(\A, \symbstateinit, \somelocs, \emptyset)$ then $\somelocs$ is reachable in $\valuate{\A}{\pval}$;
		\item Completeness: For all~$\pval$, if $\somelocs$ is reachable in $\valuate{\A}{\pval}$ then $\pval \in \EEFbar(\A, \symbstateinit, \somelocs, \emptyset)$.
	\end{enumerate}
}

\begin{proposition}\label{prop:EEFbar:correctness}
	\enonceTheoremEEFbarCorrectness{}
\end{proposition}

\begin{lemma}\label{lemma:unnamed:A}
The bounded valuation $\valuateLU{\A}$ of a PTA~$\A$
guarantees for each constraint in the model to give the greatest  possible constant bound for all valuations in the set of bounded parameters of~$\A$.
\end{lemma}

\newcommand{\prooflemmaunnamedA}{
\begin{proof}
In any given guard, as each upper bounded parameter of positive sign is set to its upper bound and each lower bounded parameter of negative sign is set to its lower bound, there can be no other valuation of bounded parameters such that any guard or invariant displays a greater constant part.
\end{proof}
}

\prooflemmaunnamedA{}

Note that  $\valuateLU{\A}$ might not even be in the 
set of PTA obtained when setting values for bounded parameters%
, as it is possible that a given parameter is replaced by its lower bound in some guard, and by its upper bound in some other.
It guarantees, however, that the value of the constant bound for any of the PTA obtained by valuating bounded parameters is no greater than $\constantBoundLUplus$.

We can proceed with the proof of \cref{prop:EEFbar:correctness}:

\begin{proof}
Let $\A'$ be any bounded valuation%
of~$\A$.
By definition%
, $\A'$ is either an L-PTA or a U-PTA.
From \cref{lemma:unnamed:A}, we know that $\constantBoundLUplus$ is greater than the constant bound of~$\A'$.
By \cref{prop:EEFhat:correctness}, we know that the extrapolation of~$\A'$ is sound and complete when defining $\LargestC(\clock)$ as the maximal constant of clock~$\clock$ when bounding all unbounded parameters with~$\overLargestP$.
As $\constantBoundLUplus > \overLargestP$%
, the extrapolation is still sound and complete for any bounded valuation%
of~$\A$.

\end{proof}
\subsection{Partial $\vec{\LargestC}$-extrapolation on general PTAs}\label{ss:partial-PTAs}

Finally, it is possible to perform a \emph{partial} extrapolation on any PTA~$\A$, by extrapolating only the clocks that are only compared to the set of bounded parameters $\Param_{bound}$ of~$\A$.
That is, for a given guard or invariant $\guard$ of the form $\clock \compOp \sum_{1 \leq i \leq \ParamCard} \alpha_i \param_i + z$, the maximum value $\maxCg(\guard) = \sum_{1 \leq i \leq \ParamCard} \alpha_i \gamma_i + z$ where
\begin{ienumerate}
	\item $\gamma_i = \PDomain^-(\param_i)$ if $\alpha_i < 0$,
	\item $\gamma_i = \PDomain^+(\param_i)$ if $\alpha_i > 0$, and
	\item $\gamma_i = 0$ otherwise.
\end{ienumerate}
Note that $\gamma_i$ may be $\infty$ or $-\infty$ if $\param_i$ is not an unbounded parameter.
As a result, the maximal constant of any clock $\clock_i \in \Clock$ compared to unbounded parameter is equal to~$\infty$.
Therefore, $\LargestC(\clock_i) \in \vec{\LargestC} = \infty$---which amounts to never applying extrapolation on $\clock_i$.
\ShortVersion{A (simple) formal result is given in \cite{AA22report}.}
\LongVersion{

	Let $\Clock_b$ denote the set of clocks compared to no unbounded parameter (\ie{} compared in guards and invariants only to constants or bounded parameters).

	Let $\vec{\LargestC_b} = \{ \LargestC(\clock_1), \dots, \LargestC(\clock_\ClockCard) \}$ such that $\LargestC(\clock_i)$ is the maximal constant of clock~$\clock_i$ (\ie{} $\infty$ if $\clock_i \notin \Clock_b$).
	Let $\Ext{\LargestC}{\Clock_b}$ denote the composition (in any order) of all $\Ext{\LargestC(\clock)}{\clock}$, for all $\clock \in \Clock_b$.
	Let \EEFp{} (``\stylealgo{p}'' stands for ``partial'') denote the modification of \EEF{} where $\Ext{\LargestC}{\Clock}$ is replaced with $\Ext{\LargestC}{\Clock_b}$ (\cref{algo:EEF:recursion} in \cref{algo:EEF}).

	\begin{proposition}\label{prop:EEFp:correctness}
		Let~$\A$ be a PTA with initial symbolic state~$\symbstateinit$, and $\somelocs \subseteq \Loc$ a set of target locations.
		Assume $\EEFp(\A, \symbstateinit, \somelocs, \emptyset)$ terminates.
		We have:
		\begin{enumerate}
			\item Soundness: If $\pval \in \EEFp(\A, \symbstateinit, \somelocs, \emptyset)$ then $\somelocs$ is reachable in $\valuate{\A}{\pval}$;
			\item Completeness: For all~$\pval$, if $\somelocs$ is reachable in $\valuate{\A}{\pval}$ then $\pval \in \EEFp(\A, \symbstateinit, \somelocs, \emptyset)$.
		\end{enumerate}
	\end{proposition}
	\begin{proof}
	The proof is the same as for \cref{prop:EEFvect:correctness}.
	\end{proof}

}
\begin{example}
	In \cref{figure:partial_ext}\LongVersion{ (with $\paramp$ being unbounded), which is a variation of \cref{figure:bounded_pta:pta} where $\paramp$ is now equal to $\clockx$ in the transition to $\loc_1$}, $\clockx$ is compared to\LongVersion{ the unbounded parameter}~$\paramp$ which is neither a lower bound nor an upper bound parameter.
	Therefore, this PTA is not in any of the previous classes on which it is possible to compute a constant bound.
	However, we can apply a partial extrapolation, \ie{} the extrapolation is only applied on~$\clocky$, for which there exists a maximal constant $\maxC^{\clocky}(\A) < \infty$.
	The analysis using \imitator{} returns quickly (in $< 0.1\,s$)
	the expected result $\paramp \in [0,1]$, while it cannot be solved with a standard exploration (\ie{} the algorithm would not terminate).
\end{example}

\LongVersion{%
	Of course, we have even less guarantee of termination in the case where only some clocks are extrapolated, but this can still help termination when comparing to the case without any extrapolation.
}

\section{Beyond reachability in bPTA+L and bPTA+U}\label{section:beyond}

We saw in \cref{section:M-Ex_unbounded} that it was possible to apply extrapolation on \LongVersion{unbounded PTAs, thanks to a result from \cite{BlT09}, notably }unbounded L-PTAs and U-PTAs with additional bounded parameters.
However, we only proved correctness \LongVersion{of this method }for reachability properties.
In this section, we study liveness\LongVersion{ and trace preservation properties}.

\LongVersion{\subsection{Liveness}\label{ss:liveness}}

In the context of unbounded parameters, the $\widehat{\LargestC}$-extrapolation cannot be used directly to check liveness properties, as it might produce false positives.
The U-PTA in \cref{figure:PTA:IM} exemplifies why the parametric extrapolation is not correct for cycle synthesis on unbounded PTAs. 
With this automaton, the state space is infinite with $\clocky$ growing without bound: after $i$ loops, we have $\clocky = \clockx + i \leq \paramp$.
The expected result of a cycle synthesis is $\KFalse$ (no valuation yields a cycle), but an exploration of the state space would not terminate.
If we try applying the $\widehat{\LargestC}$-extrapolation, we obtain $\LargestC(x)=1$ and $\LargestC(y)=522$ as greatest constants, computed using~$\overLargestP$ (\cref{ss:L-U}).
After $522$ loops, the valuation of $\clocky$ can be greater than $\LargestC(y)$, and we obtain a self-looping state where $\clocky > 522$ and $\paramp > 523$.
As a result, the $\widehat{\LargestC}$-extrapolation will synthesize a cycle for $\paramp > 523$, while there should be none. 
This behavior is due to the invariant $\clocky \leq \paramp$ being removed by the cylindrification of clock $\clocky$.
Note that this is not possible with bounded parameters (or general TAs) because any invariant $\clocky \leq t$, with $t$ a given constant, would necessarily contradict the constraint $\clocky > \LargestC$.
Indeed, $\LargestC$ being by definition the greatest constant of clock $\clocky$, $\LargestC \geq t$ and thus $\clocky > \LargestC \cap \clocky \leq t = \emptyset$.

\LongVersion{%
	Observe that the model in \cref{figure:PTA:IM} is a U-PTA.
	From~\cite[Theorem~6]{BlT09}, we know that there exists a maximal constant (similar to our~$\overLargestP$ computed in \cref{ss:L-U}) such that there exists no accepting cycle for any parameter valuation whenever the TA obtained from the U-PTA by valuating its parameters with~$\overLargestP$ yields no accepting cycle.
	This is not a contradiction with our example: in our method, we do not only use~$\overLargestP$ to valuate parameters, but we also apply extrapolation, which involves cylindrification (\cref{definition:x-extrapolation}).
	This is the cylindrification operator which is responsible for the incorrectness of the extrapolation.
}

A solution to fix that issue is to ensure the invariant is not ignored, by bounding~$\paramp$ by the constant~$\overLargestP$ ($522$ in this case).
In general, bounding all parameters by~$\overLargestP$ ensures no false positive are present, 
but might include false negative in the form of upper bounds (those we introduced to bound the parameters).
However, we know from~\cite[Theorems~3 and~6]{BlT09} that in an L-PTA or a U-PTA, if there is an infinite accepting run for a parameter valuation $\pval$ with $\pval(\paramp) \geq \overLargestP$, then this run exists for all valuations $\pval'$ with $\pval(\paramp) \geq \overLargestP$.
Therefore, in a U-PTA, the upper bound on~$\paramp$ can be removed on any results that contains ``$\paramp = \overLargestP$''.
This method can be applied on the classes of models on which we have defined a extrapolation using the constant bound $\overLargestP$ (\ie{} bPTA+L and{} bPTA+U).

In the case of our example from \cref{figure:PTA:IM}, this means constraining the model with $\paramp \leq 522$.
As a result, the $\widehat{\LargestC}$-extrapolation will synthesize no cycles, which is correct.
Now, imagine a model with the same constant bound over parameter $\overLargestP = 522$, but such that the expected result is $400 < \paramp$.
The $\widehat{\LargestC}$-extrapolation on the constrained model will synthesize $400 < \paramp \leq 522$---which contains $\paramp = 522$.
We can then remove the upper bound on $\paramp$ and obtain the correct result\LongVersion{, \ie{} $400 < \paramp$}.

\section{Experiments} \label{section:experiments}

We implemented 
\LongVersion{all aforementioned extrapolations}\ShortVersion{the aforementioned extrapolation}
 in \imitator{}~\cite{Andre21}; all operations on parametric zones are computed by polyhedra\LongVersion{l} operations, using PPL~\cite{BHZ08}.
We consider the full class of PTAs, over (potentially unbounded) rational-valued parameters.
We applied the extrapolation on the bPTA+L/bPTA+U subclass from \cref{ss:bPTA+L-U} when it was possible, and the partial $\vec{\LargestC}$-extrapolation from \cref{ss:partial-PTAs} otherwise (\ie{} extrapolation is applied to each clock whenever possible), %
to a library of standard PTA benchmarks~\cite{AMP21}.
\ShortVersion{We used}\LongVersion{Experiments were performed using} an Intel Core i5-4690K with \LongVersion{a clock rate of }4\,GHz.\footnote{%
	Source, benchmarks, raw results and full table are available \LongVersion{on the long-term archiving platform Zenodo }at \href{https://doi.org/10.5281/zenodo.5824264}{\nolinkurl{doi.org/10.5281/zenodo.5824264}}.
	We used a fork of \imitator{} 3.1 ``Cheese Artichoke'' extended with extrapolation functions (exact version: \href{https://github.com/imitator-model-checker/imitator/releases/tag/v3.1.0\%2Bextrapolation}{\texttt{v3.1.0+extrapolation}}).
}

We tabulate our results in 
\cref{table:experiments}%
.
The first and main outcome is the two lines for ``all models'' (in bold): on the entire benchmark set (119~models and 177~properties%
), the average execution time is 954\,s without extrapolation, and 824\,s with;
in addition, the normalized average (always taking 1 for the slowest of both algorithms and rescaling the second one accordingly) is 0.89 without and 0.91 with.
Both metrics are complementary, as the average favors models with large verification times, while normalized average gives the same weight to all models, including those of very small verification times.
The outcome is that the extrapolation decreases the average time by 14\,\%, and increases the normalized average time by 1.5\,\%, which remains near-to-negligible.
\LongVersion{%
	On the larger models ($> 5$\,s%
	), extrapolation allows for a very similar decrease of~14\,\% in average, and even a small decrease of 0.6\,\% for the normalized average time.
}

We only tabulate in 
\cref{table:experiments} %
results with the most significant difference, \ie{} with a gap of more than 1\,s with a ratio $\frac{\textit{min}}{\textit{max}} > 2$ (and only one property per model).
Put it differently, other models show little difference between both versions.
``\propEF{}'' %
denotes reachability %
synthesis%
;
``\propCycle{}'' %
denotes the synthesis of valuations leading to at least one infinite run%
.

Recall that, even on the most restrictive syntactic subclass of PTAs we considered (L-PTAs and U-PTAs), synthesis over rational-valued parameters is intractable, and therefore our algorithms (including with extrapolation) come with no guarantee of termination.
On the entire benchmarks set, 39~properties (over 33~models) do not terminate without extrapolation; this figure reduces to 33~properties (over~29 models) when applying extrapolation.
(No analysis terminating without extrapolation would lead to non-termination when adding extrapolation.)

On the models where there is a significant difference between with and without extrapolation, tabulated in \cref{table:experiments},
the extrapolation is sometimes significantly faster%
, sometimes significantly slower%
.
Most importantly, extrapolation allows termination of some so far unsolvable models%
.
The slower cases are due to the fact that our implementation in \imitator{} needs to keep each symbolic state \emph{convex}---this is required by the internal polyhedral structure.
Therefore, when a clock is extrapolated, this increases the number of states in the state space (a given extrapolated symbolic state can be potentially split into up to $2^{|\Clock|}$ new symbolic states via a single outgoing transition).
\begin{table}[tb]
\centering
\caption{%
	Execution times\LongVersion{ (in seconds)} for our experiments.
	\TO{} denotes an execution unfinished after 3,600 seconds.
	(We \LongVersion{therefore }use this value for means computation.)
	Normalized mean is the ratio to the worst execution times.
	Cells color represents the difference in performance for a given row: the lighter the better.
}
{\scriptsize
\setlength{\tabcolsep}{2pt} %
\begin{tabular}{ |l|l|r|r|r| }
\hline
\rowHeader{}
Model & Property & No extrapolation (s) & $\overline{\LargestC}$-extrapolation (s) \\
\hline
\stylebench{FischerPS08-4} & \propAGnot{} & \cellcolor{orange}10.6  & \cellcolor{orange!45}4.8 \\
\hline
\stylebench{FMTV\_2} & \propEF{} & \cellcolor{orange!33}0.7  & \cellcolor{orange}2.3 \\
\hline
\stylebench{fischerPAT3} & \propAGnot{} & \cellcolor{orange}1.9 & \cellcolor{orange!46}0.8 \\
\hline
\stylebench{SLAF14\_5} & \propAGnot{} & \cellcolor{orange!18}12.6  & \cellcolor{orange}74.4 \\
\hline
\stylebench{spsmall} & \propAGnot{} & \cellcolor{orange!2}0.4  & \cellcolor{orange}19.3 \\
\hline
\stylebench{SSLAF13\_test2} & \propAGnot{} & \cellcolor{orange}2869.8 & \cellcolor{orange!48}1399.1 \\
\hline
\stylebench{synthRplus} & \propEF{} & \cellTO{} & \cellcolor{orange!0}0.2 \\ %
\hline
\hline
\stylebench{Cycle1} & \propCycle{} & \cellTO{} &  \cellcolor{orange!0}0.001 \\
\hline
\stylebench{infinite-5} & \propCycle{} & \cellTO{} & \cellcolor{orange!0}0.006 \\
\hline
\stylebench{infinite-5\_6} & \propCycle{} & \cellTO{} & \cellcolor{orange!0}0.004 \\
\hline
\stylebench{exU\_noloop} & \propAccCycle{} & \cellcolor{orange!14}1.1 &  \cellcolor{orange}7.7 \\
\hline
\hline
\multicolumn{2}{|l|}{Mean (models from \cref{table:experiments} only)} & \cellcolor{orange}1572.5 & \cellcolor{orange!9}137.1 \\
\hline
\multicolumn{2}{|l|}{Normalized mean (models from \cref{table:experiments} only)} & \cellcolor{orange}0.697 & \cellcolor{orange!70}0.490 \\
\hline
\multicolumn{2}{|l|}{Mean (all models)} & \bfseries{}\cellcolor{orange}954.4 & \bfseries{}\cellcolor{orange!86}823.8 \\
\hline
\multicolumn{2}{|l|}{Normalized mean (all models)} & \bfseries{}\cellcolor{orange}0.891 & \bfseries{}\cellcolor{orange}0.905 \\
\hline
\end{tabular}
}

\label{table:experiments}
\end{table}

\LongVersion{%
	These experiments highlight the main drawback of our implementation, that is, extrapolated symbolic states have to be split into convex sub-states, sometimes ended up doing more computation in the process than without any extrapolation.
	\LongVersion{%
		(We will discuss it in the conclusion.)
	}%
	Despite that drawback, the extrapolation can still significantly decrease computation time.
	Furthermore, a main benefit of our extrapolation is that it can lead to a better termination, allowing to turn infinite state spaces into finite ones; this allows us to solve previously unsolvable benchmarks (\stylebench{synthRplus}, \stylebench{Cycle1}, \stylebench{infinite-5}, \stylebench{infinite-5\_6}).

}

All in all, our experiments suggest that, despite a few models (tabulated in \cref{table:experiments}) where the presence or absence of extrapolation has a significant difference of execution time, adding extrapolation remains overall harmless, with even an average decrease of 14\,\% in the execution time.
Most importantly, it allows to solve so far unsolvable benchmarks---which we consider as the main outcome.
This suggests to use extrapolation by default for \LongVersion{parameter }synthesis in PTAs\LongVersion{ using \imitator{}}.

\section{Conclusion and perspectives}\label{section:conclusion}

\LongVersion{
\subsection{Conclusion}
}

\LongVersion{In this paper, we}\ShortVersion{We} proposed several definitions of zone extrapolation for parametric \LongVersion{timed automata}\ShortVersion{TAs}.
\LongVersion{%
	We notably improve the parametric $\LargestC$-extrapolation from~\cite{ALR15} by allowing each clock to have its own bound and combining it with results from \cite{BlT09} in order to address unbounded subclasses of PTAs.
}%
We proposed a first implementation (in \imitator{}), and showed that, while extrapolation is harmless for most models, it can also decrease the computation time of larger models and, most importantly, can lead to termination (\LongVersion{with }exact synthesis) of previously unsolvable benchmarks.
Considering the difficulty of parameter synthesis for timed models, we consider it a non-trivial and promising step.

\LongVersion{
\subsection{Future works}
We now discuss future works.
}

A \LongVersion{main }limitation of our implementation \LongVersion{in \imitator{} }(discussed in \cref{section:experiments}) is that it only handles \emph{convex} \LongVersion{parametric }zones.
Using the non-convex polyhedral structures offered by PPL~\cite{BHZ08} may dramatically reduce the number of symbolic states.
However, they are much more costly than their convex counterparts---this should be experimentally compared.

Another perspective \LongVersion{on implementation }concerns the computation of the constant bounds $\overLargestP$, for which one needs to compute the number $R$ of clock regions.
Our current implementation uses its over-approximation~$\widehat{R}$.
Computing the actual number of clock regions before applying the extrapolation may considerably reduce the analysis time for larger models.

\LongVersion{%
	The main limitation of parametric extrapolation is that termination of synthesis for PTAs cannot be guaranteed, even for bounded PTAs.
	Although the motivation behind extrapolation is to replace infinite sequences by cycles, this is not possible for parameters converging towards a constant.
	A perspective would be to exhibit a subclass of PTAs for which it is possible to extrapolate on parameters themselves the constant towards which they converge.
}

Finally, we plan to go beyond this work by adapting the $\LargestCL\LargestCU$-extrapolation from~\cite{BBLP06} to PTAs, a theoretically coarser abstraction for which implementation is not trivial.
Algorithms from~\cite{HSW16} may prove useful to this purpose.

\section*{Acknowledgements}
The authors would like to thank
the reviewers for their comments,
and
Dylan Marinho for his help in providing the models and automation tools that were used for the benchmarking presented in this paper.

\ifdefined\VersionLong
	\newcommand{\CCIS}{Communications in Computer and Information Science}
	\newcommand{\ENTCS}{Electronic Notes in Theoretical Computer Science}
	\newcommand{\FAC}{Formal Aspects of Computing}
	\newcommand{\FundInf}{Fundamenta Informaticae}
	\newcommand{\FMSD}{Formal Methods in System Design}
	\newcommand{\IJFCS}{International Journal of Foundations of Computer Science}
	\newcommand{\IJSSE}{International Journal of Secure Software Engineering}
	\newcommand{\IPL}{Information Processing Letters}
	\newcommand{\JAIR}{Journal of Artificial Intelligence Research}
	\newcommand{\JLAP}{Journal of Logic and Algebraic Programming}
	\newcommand{\JLAMP}{Journal of Logical and Algebraic Methods in Programming} %
	\newcommand{\JLC}{Journal of Logic and Computation}
	\newcommand{\LMCS}{Logical Methods in Computer Science}
	\newcommand{\LNCS}{Lecture Notes in Computer Science}
	\newcommand{\RESS}{Reliability Engineering \& System Safety}
	\newcommand{\STTT}{International Journal on Software Tools for Technology Transfer}
	\newcommand{\TCS}{Theoretical Computer Science}
	\newcommand{\ToPNoC}{Transactions on Petri Nets and Other Models of Concurrency}
	\newcommand{\TSE}{{IEEE} Transactions on Software Engineering}
\else
	\newcommand{\CCIS}{CCIS}
	\newcommand{\ENTCS}{ENTCS}
	\newcommand{\FAC}{FAC}
	\newcommand{\FundInf}{FI}
	\newcommand{\FMSD}{FMSD}
	\newcommand{\IJFCS}{IJFCS}
	\newcommand{\IJSSE}{IJSSE}
	\newcommand{\IPL}{IPL}
	\newcommand{\JAIR}{JAIR}
	\newcommand{\JLAP}{JLAP}
	\newcommand{\JLAMP}{JLAMP}
	\newcommand{\JLC}{JLC}
	\newcommand{\LMCS}{LMCS}
	\newcommand{\LNCS}{LNCS}
	\newcommand{\RESS}{RESS}
	\newcommand{\STTT}{STTT}
	\newcommand{\TCS}{TCS}
	\newcommand{\ToPNoC}{ToPNoC}
	\newcommand{\TSE}{TSE}
\fi

\ifdefined\VersionLong
	\renewcommand*{\bibfont}{\small}
	\printbibliography[title={References}]
\else
	\bibliographystyle{splncs04} %
	\bibliography{extrapolation}
\fi

\end{document}